\documentclass[a4paper,11pt]{article}

\usepackage{graphicx}
\usepackage{amsmath,amssymb}
\usepackage{algorithm} 
\usepackage{algpseudocode} 
\usepackage{color}
\usepackage{tikz}
\usepackage[normalem]{ulem}
\usepackage{url}
\usepackage{xspace}
\usepackage[colorlinks=false,hidelinks]{hyperref}

\newcommand{\MS}{{\mathcal S}}

\newcommand{\ML}{{\mathcal L}}

\newcommand{\MX}{{\mathcal X}}
\newcommand{\MY}{{\mathcal Y}}

\newcommand{\bbZ}{{\mathbb{Z}}}

\newcommand{\NIL}{{\textsc{Nil}}}

\renewcommand{\eqref}[1]{(\ref{eq:#1})}
\newcommand{\secref}[1]{Section~\ref{sec:#1}}

\newcommand{\figref}[1]{Fig.~\ref{fig:#1}}

\newcommand{\lemref}[1]{Lemma~\ref{lem:#1}}
\newcommand{\propref}[1]{Proposition~\ref{prop:#1}}
\newcommand{\thmref}[1]{Theorem~\ref{thm:#1}}

\renewcommand{\algref}[1]{Algorithm~\ref{alg:#1}}

\newcommand{\lineref}[1]{line~\ref{codeline:#1}}

\newcommand{\outNeigh}[2]{{\Gamma^{+}_{#1}(#2)}}
\newcommand{\inNeigh}[2]{{\Gamma^{-}_{#1}(#2)}}
\newcommand{\dom}{{\mathrm{Dom}}}

\newcommand{\DT}{{\mathit{DT}}}

\newcommand{\Trans}[1]{{{#1}^{\top}\!}}
\newcommand{\Ch}{\mathit{Ch}}

\newtheorem{defn}{Definition}

\newtheorem{lem}{Lemma}
\newtheorem{prop}{Proposition}
\newtheorem{thm}{Theorem}

\newcommand{\Minimal}{{Min}\xspace}
\newcommand{\Maximal}{{Max}\xspace}

\newcommand{\RemSet}{{RS}\xspace}
\newcommand{\MinRemSet}{{\Minimal\RemSet}\xspace}
\newcommand{\MinRemSets}{{{\MinRemSet}s}\xspace}

\newcommand{\ProSubSol}{{PSS}\xspace}
\newcommand{\MaxProSubSol}{{\Maximal\ProSubSol}\xspace}
\newcommand{\MaxProSubSols}{{{\MaxProSubSol}s}\xspace}

\newcommand{\MRS}{\textup{\textsc{\MinRemSet}}}
\newcommand{\MPSS}{\textup{\textsc{\MaxProSubSol}}}

\newcommand{\True}{{\textsc{True}}\xspace}
\newcommand{\False}{{\textsc{False}}\xspace}
\newcommand{\GENMRSWOUTRT}{{Gen-\MinRemSets-Without-Root}\xspace}

\newcommand{\QED}{\quad $\Box$\medskip}

\algrenewcommand{\algorithmicrequire}{\textbf{Input: }}
\algrenewcommand{\algorithmicensure}{\textbf{Output: }}

\long\def\invis#1{}

\newenvironment{proof}{\medskip
  \noindent{\scshape Proof:}}{\quad $\Box$\medskip}

\usetikzlibrary{positioning,arrows.meta,calc}

\title{SSD Set System, Graph Decomposition\\ and Hamiltonian Cycle}
\author{Kan Shota \and Kazuya Haraguchi}
\date{{\{shota.kan, haraguchi\}@amp.i.kyoto-u.ac.jp}}

\begin{document}
\maketitle
\begin{abstract}
  In this paper, we first study
  what we call Superset-Subset-Disjoint (SSD) set system.
  Based on properties of SSD set system, we derive the following (I) to (IV):
  (I) For a nonnegative integer $k$
  and a 
  graph $G=(V,E)$ with $|V|\ge2$,
  let $X_1,X_2,\dots,X_q\subsetneq V$ denote
  all maximal proper subsets of $V$ that induce $k$-edge-connected subgraphs. 
  Then at least one of (a) and (b) holds:
  (a) $\{X_1,X_2,\dots,X_q\}$ is a partition of $V$; and
  (b) $V\setminus X_1, V\setminus X_2,\dots,V\setminus X_q$ are pairwise disjoint.
  (II) For $k = 1$ and a strongly-connected digraph $G$,
  whether $V$ is in (a) and/or (b) can be decided in $O(n+m)$ time
  and we can generate all such $X_1,X_2,\dots,X_q$
  in $O(n+m+|X_1|+|X_2|+\dots+|X_q|)$ time,
  where $n=|V|$ and $m=|E|$. 
  (III) For a digraph $G$,
  we can enumerate in linear delay all vertex subsets of $V$ 
  that induce strongly-connected subgraphs.  
  (IV) A digraph is Hamiltonian
  if there is a spanning subgraph that is 
  strongly-connected and in the case (a).   
\end{abstract}

\section{Introduction}
\label{sec:intro}

Suppose that we are given an element set $U$. 
For a family $\MS \subseteq 2^U$ of subsets of elements, 
we call $(U, \MS)$ or $\MS$ a \emph{set system}
and each subset $S\in\MS$ an \emph{$\MS$-solution}.
Let $C\subseteq U$.
If an $\MS$-solution $S$ is a proper subset of $C$, then
we call $S$ a {\em proper subset $\MS$-solution} of $C$
and write it as an {\em $\MS$-PSS}. 
We call an inclusion-wise maximal $\MS$-PSS of $C$
a {\em maximal $\MS$-PSS} of $C$
and write it as an {\em $\MS$-MaxPSS}. 
We call a nonempty subset $Y$ of $C$
an \emph{$\MS$-removable set} of $C$
if $C \setminus Y$ is an $\MS$-solution.
We write an $\MS$-removable set as an \emph{$\MS$-\RemSet}. 
We call an inclusion-wise minimal $\MS$-RS of $C$
a \emph{minimal $\MS$-removable set} of $C$
and write it as an \emph{$\MS$-\MinRemSet}. 
We denote by $\MPSS_{\MS}(C)$ the family
of all $\MS$-MaxPSSs of $C$
and by $\MRS_{\MS}(C)$ the family
of all $\MS$-MinRSs of $C$.
Observe that 
$X\in\MPSS_{\MS}(C)$ implies
$C\setminus X\in\MRS_{\MS}(C)$
and $Y\in\MRS_{\MS}(C)$ implies
$C\setminus Y\in\MPSS_{\MS}(C)$. 
We say that a subset $C$ is 
{\em $\MS$-MaxPSS-disjoint} if all $\MS$-MaxPSSs of $C$ are pairwise disjoint
and {\em $\MS$-MinRS-disjoint} if all $\MS$-MinRSs of $C$ are pairwise disjoint.
We omit the prefix $\MS$- from these notations
when it is clear from the context;
e.g., an $\MS$-solution may be simply stated as a solution.

A system $(U,\MS)$ is a \emph{Subset-Disjoint}
(\emph{SD}) \emph{system}~\cite{TH.2023}
if, for any two solutions $S,S'\in\MS$
such that $S\supsetneq S'$,
every \MinRemSet $Y$ of $S$
is either a subset of $S'$ (i.e., $Y\subseteq S'$)
or disjoint with $S'$ (i.e., $Y\cap S'=\emptyset$). 

We start this paper
by introducing what we call SSD system,
an extension of SD system. 
\begin{defn}
  A set system $(U,\MS)$ is
  a \emph{Superset-Subset-Disjoint} \emph{(SSD)}
  \emph{system}
  if, for any two solutions $S,S'\in\MS$
  such that $S\supsetneq S'$,
  every \MinRemSet $Y$ of $S$ satisfies
  at least one of the following:
  $Y$ is a superset of $S'$ $(Y\supseteq S')$;
  $Y$ is a subset of $S'$ $(Y\subseteq S')$; and
  $Y$ is disjoint with $S'$ $(Y\cap S'=\emptyset)$.   
\end{defn}

Every SD system is an SSD system,
but the converse is not true. 
For the first result,
we derive the following theorem
on SSD system. 
\begin{thm}
  \label{thm:disjoint}
  Let $(U,\MS)$ be an SSD system.
  Any solution $S\in\MS$ is
  \MaxProSubSol-disjoint and/or $S$ is \MinRemSet-disjoint. 
\end{thm}
Note that there may be a solution that is \MaxProSubSol-disjoint and \MinRemSet-disjoint at the same time.

We utilize \thmref{disjoint} to 
derive the following Theorems~\ref{thm:partition}
to \ref{thm:hamilton} on graph problems.
We denote by $\bbZ$ and $\bbZ_+$
the sets of integers and nonnegative integers, respectively.
We assume that readers are familiar with fundamental terminologies in graph theory;
e.g., \cite{Diestel.2017,We.2018}. 
Let $G=(V,E)$ be a directed graph (digraph) or
an undirected graph with a vertex set $V$
and an edge set $E$. 
We denote $n:=|V|$ and $m:=|E|$. 
For $S\subseteq V$, we denote by $G[S]$ the
\emph{subgraph induced by $S$}, that is,
$G[S]\triangleq(S,E[S])$, where $E[S]$ is defined to be
$E[S]\triangleq\{(u,v)\in E\mid u,v\in S\}$.
A subset $F\subseteq E$ is called an \emph{edge-cut} (or a \emph{cut})
if there are distinct vertices $s, t \in V$
such that no path from $s$ to $t$ exists in the subgraph $(V, E \setminus F)$,
where an empty set can be an edge-cut.
The minimum size of a cut
is called the \emph{edge-connectivity of $G$} and denoted by $\lambda(G)$,
where we define $\lambda(G)\triangleq+\infty$ if $|V|=1$. 
For $k\in\bbZ_+$, if $k\le\lambda(G)$,
then $G$ is said to be \emph{$k$-edge-connected}. 
In particular, a 1-edge-connected digraph
is called \emph{strongly-connected}. 
Being a fundamental concept of graph theory, 
edge-connectivity 
has attracted much attention of researchers
from algorithm theory and network analysis
for many years~\cite{AMO.1993,Bang-Jensen2008,CLRS.2022,S.2004,W.2019}. 

For $k\in\bbZ_+$ and a graph $G=(V,E)$,
we consider a set system $(V,\MS_{G,k})$
such that $\MS_{G,k}\triangleq\{S\subseteq V\mid G[S]\textrm{ is }k\textrm{-edge-connected}\}$.
We call $\MS_{G,k}$ a \emph{$k$-edge-connected system},
and in particular, we call it a \emph{strongly-connected system} if $G$ is directed and $k=1$. 
If $G$ is clear from the context, then we abbreviate
$\MS_{G,k}$ into $\MS_k$. 
Furthermore, for $C\subseteq V$,
we abbreviate $\MPSS_{\MS_{k}}(C)$ and $\MRS_{\MS_{k}}(C)$
into $\MPSS_{k}(C)$ and $\MRS_{k}(C)$, respectively. 
We say that $G$ is $\MS_{k}$-\MaxProSubSol-disjoint (resp., $\MS_{k}$-\MinRemSet-disjoint) if 
$V$ is $\MS_{k}$-\MaxProSubSol-disjoint (resp., $\MS_{k}$-\MinRemSet-disjoint). 
We will show that $\MS_{k}$ is an SSD system
(\lemref{k_edge_SSD})
to obtain the following theorem.

\begin{thm}
  \label{thm:partition}
  Let $k\in\bbZ_+$ and
  $G=(V,E)$ be a graph with $|V|\ge2$. 
  Then $G$ is $\MS_k$-\MaxProSubSol-disjoint and/or $\MS_k$-\MinRemSet-disjoint,
  and if $G$ is $\MS_k$-\MaxProSubSol-disjoint,
  then 
  $\MPSS_{k}(V)$ is a partition of $V$. 
\end{thm}

\thmref{partition} appears to be trivial.
When $G$ is not $k$-edge-connected (i.e., $V\notin \MS_{k}$),
it is easy to show that
$\MPSS_{k}(V):=\{X_1,X_2,\dots,X_q\}$ is a partition of $V$.
Furthermore, $X_1,X_2,\dots,X_q$ can be generated
by existing algorithms. 
For example, for $k=1$, 
textbook algorithms~\cite{CLRS.2022} generate
$X_1,X_2,\dots,X_q$ in $O(n+m)$ time,
where $X_1,X_2,\dots,X_q$ are
called connected components
(resp., strongly-connected components)
when $G$ is undirected (resp., directed). 
For a digraph, when $k=2$,
we can generate $X_1,X_2,\dots,X_q$
in $O(m^{3/2})$ time~\cite{CHILP.2017}.
For an undirected graph,
the task can be done
in $O(m \sqrt{n})$ time when $k = 3$~\cite{CHILP.2017}
and in $O(m+n^{1+O(1)})$ time when $k=\log^{O(1)}n$~\cite{SY.2023}. 
A related problem is to generate ``$k$-edge-connected components''.
For a digraph, we can generate them
in linear time when $k = 3$~\cite{FOCS.2024}.
For an undirected graph, this task is achievable in linear time 
when $k = 4 $~\cite{ESA.2021-2,ESA.2021-1}
and $k = 5$~\cite{SODA.2024}.
Note that $k$-edge-connected components differ from $X_1, X_2, \dots, X_q$
for a digraph when $k \geq 2$ and
for an undirected graph when $k \geq 3$.

The point of \thmref{partition}
is that we admit $G$ itself to be $k$-edge-connected
(i.e., $V\in\MS_{k}$).
Let us consider the problem of generating $\MPSS_{k}(V)=\{X_1,X_2,\dots,X_q\}$
when $V\in\MS_{k}$.
In this case, if $V$ is given as input,
the above-mentioned algorithms
will output $V$ as the unique solution,
which is not a \MaxProSubSol of $V$. 

The \MaxProSubSols $X_1,X_2,\dots,X_q$ of $V$
and their sizes $|X_1|,|X_2|,\dots,|X_q|$
can be used to
evaluate the robustness of a network that works based on its $k$-edge-connectivity.
Even if the network itself is $k$-edge-connected,
some terminals may fail.
It would be preferred that $G$ contains larger \MaxProSubSols
since they are the ``second'' largest $k$-edge-connected graphs
in the network.
To observe extreme examples, suppose $k=1$. 
If $G$ is a directed cycle
(i.e., the ``sparsest'' 1-edge-connected digraph), 
then every \MaxProSubSol is a single vertex
and its size is only 1. 
On the other hand, if $G$ is a complete directed graph
with the edge set $E=\{(u,v)\in V\times V\mid u\ne v\}$
(i.e., the ``densest'' 1-edge-connected digraph), 
then every \MaxProSubSol is $V\setminus\{v\}$, $v\in V$ and
its size is $n-1$.
The sizes of \MaxProSubSols can be
an alternative measure of robustness beyond the conventional edge/vertex-connectivity. 

We consider the generation problem for strongly-connected system.
We show that all \MaxProSubSols of $V$ can be generated in linear time,
as stated in the following theorem,
where dominator tree~\cite{ILS.2012} is a key tool of the algorithm. 
\begin{thm}
  \label{thm:strong}
  Given a strongly-connected digraph $G=(V,E)$,
  let us denote $\MPSS_{G,1}(V)=\{X_1,X_2,\dots,X_q\}$.
  We can decide whether
  $G$ is \MaxProSubSol-disjoint or not in $O(n+m)$ time
  and generate all $X_1,X_2,\dots,X_q$
  in $O(n+m+|X_1|+|X_2|+\dots+|X_q|)$ time. 
\end{thm}

\thmref{strong} can be used to design an efficient algorithm
for enumerating all solutions in a strongly-connected system $\MS_{G,1}$.  
An enumeration problem in general asks to output all
required solutions without duplication
and has many applications in such fields
as data mining~\cite{AIS.1993,AS.1994,Chang.2013,KK.2005} and bioinformatics~\cite{E-MEM.2014,Roux.2022,SSF.2010}. 
\emph{Delay} of an enumeration algorithm 
refers to any computation time in
the following three cases:
(i) computation time between the start of the algorithm
and the output of the first solution;
(ii) computation time between any two consecutive outputs;
and (iii) computation time between the last output
and the halt of the algorithm.
An enumeration algorithm achieves \emph{polynomial delay} (resp., \emph{linear delay})
if the delay is bounded by a polynomial (resp., a linear function)
with respect to the input size,
where polynomial delay (and also linear delay)
is among the strongest classes
of enumeration algorithms~\cite{JYP.1988}.
\begin{thm}
  \label{thm:linear_delay}
  Given a directed graph $G$,
  we can enumerate all solutions in $\MS_{G,1}$ in $O(n+m)$ delay.
\end{thm}
\thmref{linear_delay} improves the delay bound $O(n^5m)$ in Theorem~5(i) of
\cite{Haraguchi.2022}.

The last result is 
a new sufficient condition on the existence of Hamiltonian cycles,
a well-known concept of graph theory. 
A \emph{Hamiltonian cycle} in a graph
is a cycle that visits all vertices exactly once.
A graph $G$ is \emph{Hamiltonian} if there exists
a Hamiltonian cycle in $G$.
\begin{thm}
  \label{thm:hamilton}
  A digraph $G=(V,E)$ is Hamiltonian 
  if there is a subset $F\subseteq E$ such that
  the spanning subgraph $G'=(V,F)$ is
  strongly-connected and $\MS_{G',1}$-\MaxProSubSol-disjoint. 
\end{thm}

The paper is organized as follows.
We describe background and related work in \secref{related}
and prepare terminologies and notations in \secref{prel}.
We show properties of general SSD systems and $k$-edge-connected systems
in \secref{disjoint}, with proofs for Theorems~\ref{thm:disjoint} and~\ref{thm:partition}.
We then provide 
proofs for Theorems~\ref{thm:strong} to~\ref{thm:hamilton} 
in Sections~\ref{sec:strong} to~\ref{sec:hamilton}, respectively.
Technically, \thmref{hamilton} is obtained from fundamental properties of strongly-connected system.
Readers can skip Sections~\ref{sec:strong} and~\ref{sec:linear_delay}
if they are interested only in \thmref{hamilton}. 
Finally, we give concluding remarks in \secref{conc}. 

\section{Background and Related Work}
\label{sec:related}

We are motivated by the study~\cite{TH.2023} on SD system. 
Assuming an oracle that computes a \MinRemSet of a query solution, 
the authors in \cite{TH.2023} show that 
all solutions of a given SD system $(U,\MS)$ can be enumerated in linear delay
with respect to $|U|$ and the oracle running time. 
They also show
that the 2-edge/vertex-connected system\footnote{A $k$-vertex-connected system is defined analogously to a $k$-edge-connected system.} for an undirected graph is an SD system,
yielding linear-delay enumeration of all solutions in the system.
However, it is hard to apply their framework
to other subgraph enumeration problems.
For example, the strongly-connected system for a directed graph
is not an SD system; see \figref{cx_notSD} for a counterexample.
This is the main reason why we study SSD system,
and \thmref{linear_delay} in this paper is a solution to
the problem of enumerating
strongly-connected induced subgraphs
that cannot be treated by SD system. 

\begin{figure}[t!]
  \centering
  \begin{tikzpicture}[entity-vertex/.style={draw, circle,minimum width=0.25cm, ,inner sep=-1mm, thick}]
  \node[entity-vertex] (v1) at ($(90-72:1)+(-2.,-0.75)$) {};
  \node[entity-vertex] (v2) at ($(90:1)+(-2.,-0.75)$) {};
  \node[entity-vertex] (v3) at ($(90+72:1)+(-2.,-0.75)$) {};
  \node[entity-vertex] (v4) at ($(90+144:1)+(-2.,-0.75)$) {};
  \node[entity-vertex] (v5) at ($(90+216:1)+(-2.,-0.75)$) {};
  \draw[thick,-{Stealth}] (v2)--(v1);
  \draw[thick,-{Stealth}] (v3)--(v2);
  \draw[thick,-{Stealth}] (v4)--(v3);
  \draw[thick,-{Stealth}] (v5)--(v4);
  \draw[thick,-{Stealth}] (v1)--(v5);

  \node[entity-vertex] (u1) at(0, 0) {};
  \node[entity-vertex] (u2) at(0, 1) {};
  \node[entity-vertex] (u3) at(-1, 1) {};
  
  \draw[thick,-{Stealth}] (v1)--(u1);
  \draw[thick,-{Stealth}] (u2)--(u3);
  \draw[thick,-{Stealth}] (u3)--(v2);
  \draw[thick,-{Stealth}] (u1)to[out=120,in=250](u2);
  \draw[thick,-{Stealth}] (u2)to[out=-60,in=60](u1);

  \draw [black,very thick,dotted] (-2.,-0.75)circle[radius=1.25];
  \draw [red,rounded corners=5,very thick] (-0.3,-0.25)--(0.3,-0.25)--(0.3,1.25)--(-0.3,1.25)--node[left]{\small $ S' $}cycle;
  \draw[blue,rounded corners=20,very thick,dotted]($ (u1) + (0.5,-1) $)--($ (u2) + (0.5,0.5) $)--($ (u3) + (-1,0.25) $)--cycle;
  \draw ($ (u1) + (0.95,-1) $) node[blue]{\small $Y$: a MinRS of $S$};
  \draw ($ (v5)+(0.5, 0) $) node[black,right]{\small $ X = S \setminus Y $};
  \draw[cyan,rounded corners=0,very thick] (2.6, 1.75)--(-3.5, 1.75)--(-3.5, -2.25)--(2.6, -2.25)--cycle;
  \draw(-3.5, 1.75)node[cyan, below right]{\small $ S $};
\end{tikzpicture}
  \caption{Strongly-connected system is not SD; $Y$ is a \MinRemSet of $S$ and $S'$ is a PSS of $S$,
    where we see that $Y\not\subseteq S'$ and $Y\cap S'\ne\emptyset$.}
  \label{fig:cx_notSD}
\end{figure}

To obtain \thmref{linear_delay},
we first show the proof for \thmref{disjoint}
(an important property of SSD system)
and then go on to those for
Theorems~\ref{thm:partition} and \ref{thm:strong}. 
\thmref{partition} suggests the existence
of a graph decomposition based on \MaxProSubSols 
and \thmref{strong} indicates
that the decomposition
can be done in linear time for $k=1$ and a digraph.
This graph decomposition is interesting by itself
since it is unique, or we may say canonical,
with respect to the parameter $k$.  
There have been developed
various types of graph decomposition techniques
in the literature
(e.g., graph decomposition by a ``width'' parameter~\cite{R.1976,RS.1984},
combinatorial decomposition theory~\cite{C.1982,CE.1980},
modular decomposition~\cite{HMSZ.2024}). 
Recently, a canonical graph decomposition on 3-vertex-connected graphs
is studied by Carmesin and Kurkofka~\cite{CK.2023}. 
We have not found any nontrivial relationships between
our graph decomposition and existing ones.


Let us review problems in the literature
that are related to the problem
we consider for \thmref{strong} (i.e., generation of \MaxProSubSols). 
For a graph property $\Pi$,
let us call a graph a \emph{$\Pi$-graph} if it has the property $\Pi$
and a \emph{non-$\Pi$-graph} otherwise. 
Problems of generating maximal $\Pi$-subgraphs 
for a given non-$\Pi$-graph are studied in many contexts.
Connected components and 
strongly-connected components~\cite{CLRS.2022};
maximal $k$-edge-connected subgraphs~\cite{NS.2023,SY.2023}
and maximal $k$-vertex-connected subgraphs~\cite{Wen.2019}
are typical examples. 
The \emph{vertex deletion problem}
is a problem of this kind
that asks to find a maximum induced $\Pi$-subgraph  for a given non-$\Pi$-graph.
Recently, cluster~\cite{aprile2023tight}, $r$-rank~\cite{meesum2016rank} and deadlock-free~\cite{carneiro2019deadlock} are studied for $\Pi$. 
In this paper, we deal with problems of generating
maximal proper subsets $X\subsetneq V$
that induce $\Pi$-subgraphs for a given graph which may be a $\Pi$-graph,
where we consider $k$-edge-connectivity for $\Pi$. 

The \emph{critical node detection problem}~\cite{LTK.2018}
asks to find a vertex (or vertices) whose removal maximally decreases a prescribed connectivity criterion.
Strongly-connected graphs are studied in this context;
e.g., the removal of one vertex~\cite{GIP.2020};
and the removal of a vertex subset whose size is no more than a constant~\cite{BCR.2019}. 
The aim of these studies is rather to construct a data structure
that can provide the answer to a query efficiently.
The problem of finding $\MS_{1}$-\MaxProSubSols
for a given strongly-connected digraph
has not been treated in the literature.

It is a well-known $\mathcal{NP}$-complete
problem to decide whether a given graph is Hamiltonian or not~\cite{GJ.1979}.
There are a lot of studies on conditions on the existence of Hamiltonian cycles.
See \cite{Bang-Jensen2008,KO.2012} for a collection of previous results. 
Many known conditions
are based on the distribution of vertex degrees
or special graph classes such as semicomplete digraphs and
quasi-traisitive digraphs. 
\thmref{hamilton} is not among these cases. 

\section{Preliminaries}
\label{sec:prel}
For $p,q\in\bbZ$ $(p \leq q)$, let $[p,q]:=\{p,p+1,\dots,q\}$. 
For two sets $A,B$,
if $A\cap B=\emptyset$,
then we may write their union $A\cup B$ as $A\sqcup B$ (i.e., disjoint union)
in order to emphasize that they are disjoint.

\subsection{Set Systems}
\label{sec:prel_set}
Let $(U,\MS)$ be a set system.
For $C\subseteq U$, let us denote by
$\textsc{MaxSS}_\MS(C)$ the family of
inclusion-wise maximal $\MS$-solutions
among subsets of $C$. 
Observe that
$\textsc{MaxSS}_\MS(C)$ and $\MPSS_\MS(C)$ are not necessarily the same;
if $C$ is an $\MS$-solution,
then $\textsc{MaxSS}_\MS(C)=\{C\}\ne\MPSS_\MS(C)$ holds
since all solutions in $\MPSS_\MS(C)$ are proper subsets of $C$. 
Otherwise, $\textsc{MaxSS}_\MS(C)=\MPSS_\MS(C)$ holds.

For subsets $I\subseteq C\subseteq U$,
we denote $\MS(C):=\{S\in\MS\mid S\subseteq C\}$;
$\MS(C,I):=\{S\in\MS\mid I\subseteq S\subseteq C\}$; and
$\MRS_\MS(C,I):=\{Y\in\MRS_\MS(S)\mid Y\cap I=\emptyset\}$. 

As mentioned in \secref{intro}, a solution $S$ in a set system $(U, \MS)$ can be
\MaxProSubSol-disjoint and \MinRemSet-disjoint at the same time.
It is trivial that this happens when $|\MRS_{\MS}(S)|\in\{0,1\}$.
For the case of $|\MRS_{\MS}(S)|\ge2$, we show the following proposition as a fundamental property.
\begin{lem}
  \label{lem:sametime}
  Let $(U,\MS)$ be any set system and $S\in\MS$ be a solution such that $|\MRS_{\MS}(S)|\ge 2$. 
  Then $S$ is \MaxProSubSol-disjoint and \MinRemSet-disjoint
  if and only if $|\MRS_{\MS}(S)|=2$ holds and $\MRS_{\MS}(S)$ is a partition of $S$.
\end{lem}
\begin{proof}
  Let $\MRS_{\MS}(S)=\{Y_1,Y_2,\dots,Y_q\}$ $(q\ge2)$
  and $\MPSS_{\MS}(S)=\{X_1,X_2,\dots,X_q\}$ such that $X_i=S\setminus Y_i$, $i\in[1,q]$. 
  $(\Longleftarrow)$ We see that $S$ is \MinRemSet-disjoint
  since $Y_1\cap Y_2=\emptyset$. Also, by $Y_1\cup Y_2=S$,
  we have $X_1=S\setminus Y_1=Y_2$ and $X_2=S\setminus Y_2=Y_1$, indicating that $S$ is \MaxProSubSol-disjoint. 
  $(\Longrightarrow)$ For contradiction, if $q\ge3$, we have that
  $X_1\cap X_2=X_1\cap X_3=X_2\cap X_3=\emptyset$ since $S$ is \MaxProSubSol-disjoint.
  Then $Y_1\cup Y_2=Y_1\cup Y_3=Y_2\cup Y_3=S$ would hold,
  contradicting that $S$ is \MinRemSet-disjoint.   
  We have $q=2$, and it is easy to see that $\MRS_{\MS}(S)=\{Y_1,Y_2\}$ is a partition of $S$. 
\end{proof}

Although we provided the original definitions of SD and SSD systems
in \secref{intro}, let us rephrase them in terms of \MaxProSubSols for better understanding. 
\begin{itemize}
\item A system $(U,\MS)$ is an SD system if, for any two solutions $S,S'\in\MS$
  such that $S'\subsetneq S$,
  every $X\in\MPSS_\MS(S)$ satisfies 
  either $X\cup S'=S$ or $X\supseteq S'$. 
\item A system $(U,\MS)$ is an SSD system if, for any two solutions $S,S'\in\MS$
  such that $S'\subsetneq S$,
  every $X\in\MPSS_\MS(S)$ satisfies 
  at least one of the following:
  $X\cap S'=\emptyset$; $X\cup S'=S$ and $X\supseteq S'$. 
\end{itemize}
For a system $(U,\MS)$, let $S\in\MS$ be a solution and $X,X'\in\MPSS_\MS(S)$. 
If $\MS$ is an SD system, then $X\cup X'=S$ holds.
Similarly, if $\MS$ is an SSD system, at least one of $X\cup X'=S$ and $X\cap X'=\emptyset$ hold.
The converse does not hold true for both systems. 

Let us make a comparison with other set systems in the literature.
Two sets $A,B$ \emph{intersect} if
$A\cap B\ne\emptyset$, $A\setminus B\ne\emptyset$ and $B\setminus A\ne\emptyset$. 
\begin{itemize}
\item A set system $(U,\MS)$ is \emph{confluent}
  if, for any three solutions $S,S',S''\in\MS$,
  $S'' \subseteq S \cap S'$ implies $S \cup S' \in \MS$. 
\item A set system $(U,\MS)$ is \emph{laminar}
  if no two solutions $S,S'\in\MS$ intersect. 
\end{itemize}
A confluent system is not necessarily an SD system (and hence not necessarily an SSD system);
e.g., $(U,\MS)$ with $U=\{1,2,3,4,5,6\}$ and $\MS=\{\{1,2,3,4,5,6\},\{1,2,3\},\{3,4,5\}\}$.
An SD or SSD system is not necessarily confluent;
e.g., $(U,\MS)$ with $U=\{1,2,3,4,5\}$ and $\MS=\{\{1,2,3\},\{3,4,5\},\{3\}\}$.

We see that any laminar system $(U,\ML)$ is an SSD system as follows;
Let $S,S'\in\ML$ be solutions such that $S'\subsetneq S$
and $X\in\MPSS_\ML(S)$.
It holds that $X\cap S'=\emptyset$; $X\supseteq S'$ or $X\subseteq S'$.
In the last case, $X=S'$ should hold by the maximality of $X$,
and hence $X\supseteq S'$ holds.
Substituting $X=S\setminus Y$ for $Y\in\MRS_\ML(S)$, we have that
either $Y\supseteq S'$ or $Y\cap S'=\emptyset$ holds,
indicating that any laminar system is ``Superset-Disjoint'' in our
terminologies and hence an SSD system. 
A laminar system is not necessarily an SD system
and there is an SD system that is also laminar.  

The above-mentioned relationships are summarized in \figref{setsystems},
with concrete examples of set systems.
As mentioned in previous sections,
SD system is first introduced in \cite{TH.2023},
where 2-edge/vertex-connected systems
for an undirected graph are shown to be SD systems. 
Laminar system appears in the literature of
discrete mathematics/optimization, 
and in our context,
the family of ``extreme vertex subsets'' for an undirected graph~\cite{NI.2008}
would be a good example.
Confluent system is treated in the context of enumeration
algorithms~\cite{BHPW.2010,Haraguchi.2022}.
We show that $\MS_{G,k}$ for any $k\in\bbZ_+$ and undirected/directed graphs is confluent (\lemref{edge_conn})
as well as SSD (\lemref{k_edge_SSD}) in
\secref{disjoint_Sgk}.

\begin{figure}[t!]
  \centering
  \includegraphics[width=10cm]{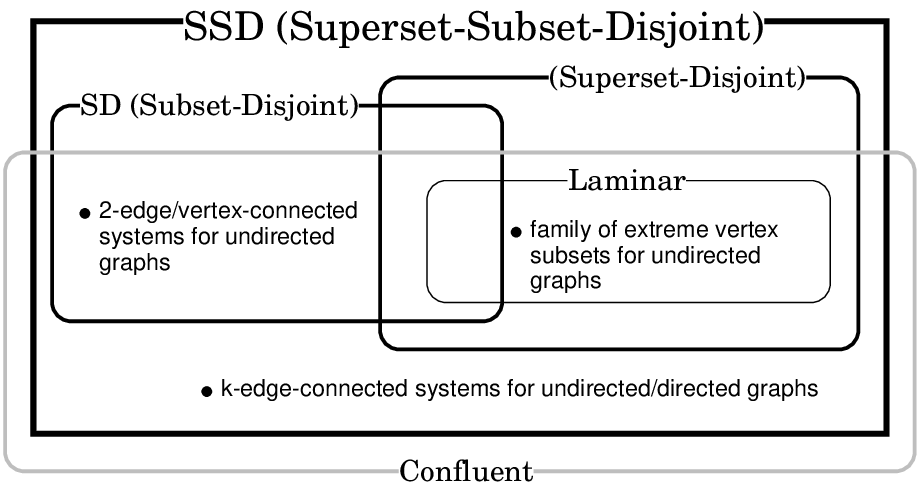}
  \caption{System classes and examples}
  \label{fig:setsystems}
\end{figure}




\subsection{Graphs}
Throughout the paper, we assume a graph to be simple and directed
except \secref{disjoint_Sgk}.
Only in \secref{disjoint_Sgk},
a graph may be undirected or may contain multiedges.
We do not deal with self-loops
in any part of the paper.

Let $G=(V,E)$ be a digraph.
For a vertex subset $S\subseteq V$, we denote by $G-S$ the induced subgraph $G[V\setminus S]$.
When $S=\{v\}$, we write $G-\{v\}$ as $G-v$ for simplicity. 
For an edge subset $F\subseteq E$,
we denote by $G-F$ the subgraph $(V,E\setminus F)$.
When $F=\{e\}$, we write $G-\{e\}$ as $G-e$. 
Let $\bar{E}:=\{(u,v)\in V\times V\mid (u,v)\notin E\}$.
For $F'\subseteq\bar{E}$, let us denote $G+F':=(V,E\cup F')$. 
When $F'=\{e'\}$, we write $G+\{e'\}$ as $G+e'$. 

A \emph{path} in $G=(V,E)$ is a sequence
of vertices $(v_1,v_2,\dots,v_{\ell+1})$, or $v_1\rightarrow v_2\rightarrow\dots\rightarrow v_{\ell+1}$, 
such that $(v_i,v_{i+1})\in E$, $i\in[1,\ell]$.
When $v_1=v_{\ell+1}$, it is called a \emph{cycle}.
If $v_1=s$ and $v_{\ell+1}=t$, we may call it an \emph{$s,t$-path}. 
The \emph{length of a path} or (\emph{a cycle})
is the number of edges in it, that is, it is $\ell$ in this case. 
A path or a cycle of length $\ell$
is called an \emph{$\ell$-path} or an \emph{$\ell$-cycle}. 
A \emph{simple path} is a path such that 
if $v_1,v_2,\dots,v_{\ell+1}$ are all-different,
and a \emph{simple cycle} is a cycle such that 
if $v_1,v_2,\dots,v_\ell$ are all-different and $v_1=v_{\ell+1}$. 
For convenience, we may represent a simple path
$v_1 \rightarrow v_2 \rightarrow \dots \rightarrow v_{\ell+1}$
by a vertex set $\{v_1, v_2, \dots, v_{\ell+1}\}$ unless no confusion arises.
A \emph{Hamiltonian path} in $G$ is a simple $(n-1)$-path
and a \emph{Hamiltonian cycle} in $G$
is a simple $n$-cycle.

An undirected graph (resp., a digraph) $G=(V,E)$
is \emph{connected} (resp., \emph{strongly-connected})
if there is an $s,t$-path for every distinct vertices $s,t\in V$.
A vertex subset $S\subseteq V$, or $G[S]$,
is a \emph{connected component} (resp., \emph{strongly-connected component})
if $S$ is inclusion-wise maximal such that $G[S]$ is connected (resp., strongly-connected). 
An \emph{articulation point}
(resp., a \emph{strong articulation point}) of $G$
is a vertex $v \in V$ such that
the number of connected components
(resp., strongly-connected components)
of $G - v$ is greater than that of $G$.

Let $G=(V,E)$ be a digraph. 
For a vertex $v \in V$, we denote by $\outNeigh{G}{v}$ and $\inNeigh{G}{v}$
the out-neighborhood and in-neighborhood of $v$ in $G$, respectively. 
For a subset $S\subseteq V$,
let $\outNeigh{G}{S} :=(\bigcup_{v\in S}\outNeigh{G}{v})\setminus S$ and
$\inNeigh{G}{S} :=(\bigcup_{v\in S}\inNeigh{G}{v})\setminus S$. 
The \emph{transpose of $G$} is defined to be
a digraph $G^{\top}=(V, E^{\top})$, where $E^{\top} := \{(v,u) \mid (u,v) \in E\}$.




\subsection{Dominator Trees}
A \emph{flow graph} with a start vertex $s$, denoted by $G_s$, 
is a directed graph such that every vertex is reachable from $s$.
A strongly-connected graph can be regarded as a flow graph
by taking any vertex as the start vertex. 

For a strongly-connected digraph $G=(V,E)$,
let us designate a start vertex $s\in V$. 
There is an $s,v$-path and $v,s$-path for every $v\in V$. 
We say that a vertex $u$ \emph{dominates} $v$,
or equivalently, that $u$ is a \emph{dominator of $v$} 
if every $s,v$-path visits $u$. 
We denote by $\dom(v)$ the set of dominators of $v$,
where $s,v\in \dom(v)$ holds for all $v\in V$.
Let us define a binary relation $\textrm{R}$ on $V$
such that $u\mathrm{R}v$ $\Leftrightarrow$ $u\in\dom(v)$. 
The binary relation $\mathrm{R}$ is reflexive and transitive.
The transitive reduction 
of $\mathrm{R}$
forms a tree rooted at $s$, which is called the \emph{dominator tree}
(\emph{of $G_s$})~\cite{ILS.2012}, where,
for $u, v \in V$, $u$ is the parent of $v$ if and only if
$\dom(v) = \dom(u) \cup \{v\}$.
We denote the dominator tree by $\DT_s$.
For $v\in V$, we denote by $\pi_s(v)$ the parent of $v$ in $\DT_s$,
where we define $\pi_s(s):=\NIL$.
We denote by $\Ch_s(v)$ the set of children of $v$ in $\DT_s$.
We will also consider the dominator tree of the transpose $G^{\top}_s$
and denote it by $\DT^{\top}_s$.
For $v\in V$, we denote the parent of $v$ in $\DT^{\top}_s$
by $\pi^{\top}_s(v)$,
where we define $\pi^{\top}_s(s):=\NIL$, 
and the set of children of $v$ by $\Ch^{\top}_s(v)$.


\section{Properties of SSD Systems}
\label{sec:disjoint}

In this section, we present various properties
of a general SSD set system (\secref{disjoint_SSD})
and a $k$-edge-connected system $\MS_{k}$ (\secref{disjoint_Sgk}) for
given graph $G=(V,E)$ and nonnegative integer $k\in\bbZ_+$.
We provide proofs for Theorems~\ref{thm:disjoint} and \ref{thm:partition} along this context. 

\subsection{General SSD Systems}
\label{sec:disjoint_SSD}

The following \lemref{SSD_basic} is enough for proving \thmref{disjoint}. 
\begin{lem}
  \label{lem:SSD_basic}
  For an SSD system $(U, \MS)$,
  let $S\in\MS$ be a solution such that $|\MPSS_\MS(S)|\ge2$. 
  For $X,X'\in\MPSS_\MS(S)$ $(X\ne X')$,
  \begin{itemize}
  \item[\rm (i)] if $X\cap X'\ne\emptyset$, then
    the union of any two sets in $\MPSS_\MS(S)$ is $S$; and 
  \item[\rm (ii)] if $X\cap X'=\emptyset$, then
    $S$ is \MaxProSubSol-disjoint. 
  \end{itemize}
\end{lem}
\begin{proof}
  Observe that $X\cup X'=S$ holds if $X\cap X'\ne\emptyset$ by SSD property.
  The lemma is obvious when $|\MPSS_\MS(S)|=2$.   
  Suppose $|\MPSS_\MS(S)|\ge 3$.
  Let $Z\in\MPSS_\MS(S)$, $Z\notin\{X,X'\}$.

  \noindent
  (i) We have $Z\cap X\ne\emptyset$ since otherwise $Z\subseteq S\setminus X\subsetneq X'$ would hold, contradicting the maximality of $Z$.
  Then $Z\cup X=S$ holds. By induction, we see that any $Z',Z''\in\MPSS_\MS(S)$ $(Z'\ne Z'')$
  satisfies $Z'\cup Z''=S$. 

  \noindent
  (ii) If $Z\cap X\ne\emptyset$, then $Z\cup X=S$ holds by (i).
  We would have $Z\supsetneq S\setminus X\supseteq X'$,
  contradicting the maximality of $X'$.
  By induction, we see that any $Z',Z''\in\MPSS_\MS(S)$ $(Z'\ne Z'')$
  satisfies $Z'\cap Z''=\emptyset$. 
\end{proof}

\paragraph{Proof for \thmref{disjoint}.}
The theorem holds for the case of $|\MPSS_\MS(S)|\le1$.
Suppose $|\MPSS_\MS(S)|\ge 2$.
By \lemref{SSD_basic},
if $S$ is not \MaxProSubSol-disjoint, 
then $X\cup X'=S$ holds for any two $X,X'\in\MPSS_\MS(S)$ $(X\ne X')$.
Let $Y,Y'$ denote the \MinRemSets of $S$ such that $X=S\setminus Y$ and $X'=S\setminus Y'$,
respectively. Then $Y\cap Y'=(S\setminus X)\cap(S\setminus X')=S\setminus(X\cup X')=\emptyset$ holds. 
\QED

The rest of this subsection is devoted to
various properties of SSD systems. 
The following \lemref{SSD_basic_Y_ver} is a restatement of \lemref{SSD_basic} in terms of {\MinRemSet}s.

\begin{lem}
  \label{lem:SSD_basic_Y_ver}
  For an SSD system $(U, \MS)$,
  let $S\in\MS$ be a solution such that $|\MRS_\MS(S)|\ge2$. 
  For $Y,Y'\in\MRS_\MS(S)$ $(Y\ne Y')$,
  \begin{itemize}
  \item[\rm (i)] if $Y\cup Y'\subsetneq S$, then
    $S$ is \MinRemSet-disjoint; and
  \item[\rm (ii)] if $Y\cup Y'=S$, then    
    the union of any two sets in $\MRS_\MS(S)$ is $S$. 
  \end{itemize}
\end{lem}
\begin{proof}
  Let $X,X'$ denote the \MaxProSubSols of $S$ such that $X=S\setminus Y$ and $X'=S\setminus Y'$.

  \noindent
  (i) We have $S\supsetneq Y\cup Y'=(S\setminus X)\cup(S\setminus X')=S\setminus(X\cap X')$,
  indicating that $X\cap X'\ne\emptyset$.
  By \lemref{SSD_basic}(i), for any $Z,Z'\in\MPSS_\MS(S)$, 
  we have $Z\cup Z'=S$.
  Let $W=S\setminus Z$ and $W'=S\setminus Z'$.
  We have $W\cap W'=(S\setminus Z)\cap(S\setminus Z')=S\setminus(Z\cup Z')=\emptyset$. 

  \noindent
  (ii) Similarly, we see that $X\cap X'=\emptyset$ holds.
  By \lemref{SSD_basic}(ii), for any $Z,Z'\in\MPSS_\MS(S)$, 
  we have $Z\cap Z'=\emptyset$. Again, let $W=S\setminus Z$ and $W'=S\setminus Z'$. 
  We see that $W\cup W'=(S\setminus Z)\cup(S\setminus Z')=S\setminus(Z\cap Z')=S$. 
\end{proof}

The following \lemref{SSD_basic_with_Y} is the counterpart of \lemref{SSD_basic}
concerning \MinRemSets.
\begin{lem}
\label{lem:SSD_basic_with_Y}
For an SSD system $(U, \MS)$,
  let $S\in\MS$ be a solution such that $|\MRS_\MS(S)|\ge2$. 
  For $Y,Y'\in\MRS_\MS(S)$ $(Y\ne Y')$,
  \begin{itemize}
  \item[\rm (i)] if $Y\cap Y'\neq \emptyset$, then
    the union of any two sets in $\MRS_\MS(S)$ is $S$; and
  \item[\rm (ii)] if $Y\cap Y'=\emptyset$, then    
    $S$ is \MinRemSet-disjoint. 
    \end{itemize}
\end{lem}
\begin{proof}
  (i) By \thmref{disjoint}, $S$ is \MaxProSubSol-disjoint.
  We have, for any $W, W' \in \MRS_\MS(S)$, $W \cup W' = S$
  since $(S\setminus W) \cap (S\setminus W') = S \setminus (W \cup W') = \emptyset$.

  \noindent
  (ii) If $S$ is not \MaxProSubSol-disjoint, then we are done by \thmref{disjoint}.
  Suppose that $S$ is \MaxProSubSol-disjoint and
  let $X,X'$ denote the \MaxProSubSols of $S$ such that $X=S\setminus Y$ and $X'=S\setminus Y'$.
  We have $X \cup X' = S$
  since $Y \cap Y' = (S \setminus X) \cap (S \setminus X') = S \setminus (X \cup X') = \emptyset$.
  If $|\MPSS_\MS(S)| \geq 3$, then
  at least one of $X \cap Z \neq \emptyset$ and $X' \cap Z \neq \emptyset$ would hold
  for any $Z \in \MPSS_\MS(S) \setminus \{X, X'\}$,
  contradicting the \MaxProSubSol-disjointness of $S$.
  By \lemref{sametime}, $S$ is \MinRemSet-disjoint.
\end{proof}

\medskip
An interesting implication is that
whether a solution $S$ of an SSD system is
\MaxProSubSol-disjoint or not
can be decided by checking whether $X\cap X'=\emptyset$ or not
for only two \MaxProSubSols $X,X'$ of $S$ (\lemref{SSD_basic}(i)); and
whether a solution $S$ is
\MinRemSet-disjoint or not
can be decided by checking whether $Y\cap Y'=\emptyset$ or not
for only two \MinRemSets $Y,Y'$ of $S$ (\lemref{SSD_basic_with_Y}(ii)).

The following \lemref{SSD_MRS} is used in the proof for
\lemref{str_MRS_not_disjoint} in \secref{strong}.

\begin{lem}
  \label{lem:SSD_MRS}
  Let $(U,\MS)$ be an SSD system
  and $S\in\MS$ be a \MaxProSubSol-disjoint solution. 
  For every $Y\in\MRS_\MS(S)$,
  it holds that $\MPSS_\MS(S)=\textsc{MaxSS}_\MS(Y)\sqcup\{X\}$, where $X:=S\setminus Y$. 
\end{lem}
\begin{proof}
  To show $\MPSS_\MS(S)\subseteq \textsc{MaxSS}_\MS(Y)\sqcup\{X\}$,
  let $Z\in\MPSS_\MS(S)$. Either $Z=X$ or $Z\subseteq (S\setminus X)=Y$ holds
  since $S$ is \MaxProSubSol-disjoint.
  In the latter case, $Z$ is an inclusion-wise
  maximal solution among subsets of $Y$ and hence $Z\in\textsc{MaxSS}_\MS(Y)$. 
  For the converse, $X\in\MPSS_\MS(S)$ holds by the definition.
  Any $Z'\in\textsc{MaxSS}_\MS(Y)$
  is an inclusion-wise maximal solution among subsets of $Y$. 
  If there is a \MaxProSubSol $Z''$ of $S$ such that $Z''\supsetneq Z'$,
  then it would contradict the maximality of $Z'$ or
  \MaxProSubSol-disjointness of $S$. Then $Z'\in\MPSS_\MS(S)$ holds. 
\end{proof}

Before concluding this subsection,
let us show an interesting proposition
although it is not used the paper.

\begin{prop}
  \label{prop:SSD_singleton_disjoint}
  For an SSD system $(U,\MS)$, let $S\in\MS$ be a solution.
  If there is a singleton \RemSet of $S$, then
  $S$ is \MinRemSet-disjoint. 
\end{prop}
\begin{proof}
  Let $u \in S$ be an element such that $\{u\}$ is an \RemSet of $S$.
  The singleton $\{u\}$ is also a \MinRemSet since no proper subset of $\{u\}$ is an \RemSet. 
  Let $\MRS(S)=\{Y_1,Y_2,\dots,Y_q\}$ such that $Y_1=\{u\}$.
  If $q=1$, then we are done. If $q\ge 2$,
  then $Y_1\cap Y_2=\emptyset$ holds since otherwise
  $u\in Y_2$ would hold, contradicting the minimality of $Y_2$.
  By \lemref{SSD_basic_with_Y}(ii), $S$ is \MinRemSet-disjoint. 
\end{proof}

\subsection{$k$-Edge-Connected Systems}
\label{sec:disjoint_Sgk}
Let $k\in\bbZ_+$ and $G$ be a graph,
which is either undirected or directed
and may contain multiedges. 
To prove \thmref{partition},
we show that $\MS_{k}$ is confluent (\lemref{edge_conn})
and an SSD system (\lemref{k_edge_SSD}).
Although it is already shown in \cite{Haraguchi.2022}
that $\MS_k$ is confluent,
we show the proof for the self-completeness.



\begin{lem}
  \label{lem:edge_conn}
  For any finite $k\in\bbZ_+$ and graph $G=(V,E)$,
  the system $(V,\MS_{k})$ is confluent.  
\end{lem}
\begin{proof}
  If $k=0$, then we have $\MS_{k}=2^V\setminus\{\emptyset\}$,
  and it is easy to see that $\MS_{k}$ is confluent. 
  Suppose $k\ge1$. 
  Let $S,S'\in\MS_{k}$ be any solutions such that $S\cap S'\ne\emptyset$.
  Each vertex $v$ in $S\cap S'$ is a solution
  since the edge connectivity of $G[\{v\}]$ is $+\infty>k$.   
  To prove the confluency of $\MS_{k}$,
  we show that $S\cup S'$ is a solution.
  It is clearly true when $S\subseteq S'$ or $S\supseteq S'$.
  Suppose that $S$ and $S'$ intersect.
  Let $F\subseteq E$ be a nonempty subset such that $|F|<k$.
  For every $s,t\in S$ (resp., $s,t\in S'$),
  there is an $s,t$-path in $G[S\cup S']-F$
  by the $k$-edge-connectivity of $G[S]$ (resp., $G[S']$).
  Let $s\in S\setminus S'$ and $t\in S'\setminus S$.
  Then there is a path from $s$ to $t$ in $G[S\cup S']-F$
  since, for $u\in S\cap S'$, there is a path from $s$ to $u$ in $G[S]-F$
  and there is a path from $u$ to $t$ in $G[S']-F$.
  Similarly, there is a path from $t$ to $s$ in $G[S\cup S']-F$.
  This shows that $\MS_{k}$ is confluent. 
\end{proof}



\begin{lem}
  \label{lem:k_edge_SSD}
  For any finite $k\in\bbZ_+$ and graph $G=(V,E)$,
  the system $(V,\MS_{k})$ is an SSD system.  
\end{lem}
\begin{proof}
  If $k=0$, then it is readily to see that $\MS_{k}=2^V\setminus\{\emptyset\}$
  is an SD (and hence SSD) system.
  Suppose $k\ge1$.
  For a solution $S\in\MS_{k}$, let $S'\in\MS_{k}$ be any PSS of $S$ $(S'\subsetneq S)$
  and $X \in \MPSS_{k}(S)$. 
  If $X\cap S'\ne\emptyset$ and $X\not\supseteq S'$,
  then $X\cup S'$ is a solution since $\MS_{k}$ is confluent by \lemref{edge_conn}
  and any singleton is a solution,
  where we have $X\cup S'\supsetneq X$.
  By the maximality of $X$, $X\cup S'$ must be equal to $S$. 
\end{proof}

\paragraph{Proof for \thmref{partition}.}
Note that the system $\MS_{k}$ is confluent (\lemref{edge_conn})
and an SSD system (\lemref{k_edge_SSD}). 
When $k=0$, we have $\MS_{k}=2^V\setminus\{\emptyset\}$.
Let $V=\{v_1,v_2,\dots,v_n\}$ and $X_i:=V\setminus\{v_i\}$, $i\in[1,n]$. 
Then $\MPSS_{k}(V)=\{X_1,X_2,\dots,X_n\}$ 
and $\MRS_{k}(V)=\{V\setminus X_1,V\setminus X_2,\dots,V\setminus X_n\}=\{v_1,v_2,\dots,v_n\}$ hold.
Then $V$ is \MinRemSet-disjoint.  

Suppose $k\ge1$. Each vertex $v\in V$ is a solution and
hence there is $X\in\MPSS_{k}(V)$ such that $v\in X$. 
The union of the sets in $\MPSS_{k}(V)$ is $V$.

If $V$ is not a solution,
then $V$ is \MaxProSubSol-disjoint since 
$\MS_{k}$ is confluent and each vertex itself is a solution. 
We see that $\MPSS_{k}(V)$ is a partition of $V$. 
We see that $\MPSS_{k}(V)$ is a partition of $V$.
Suppose that $V$ is a solution.
If $V$ is MaxPSS-disjoint,
then $\MPSS_{k}(V)$ is a partition of $V$.
Otherwise, $V$ is MinRS-disjoint by \thmref{disjoint}. 
\QED

\medskip
The following \lemref{num_MRSs} is about the number of \MinRemSets in a solution
in a $k$-edge-connected system
and will be used in later sections.
\lemref{MRS_exist_not_v} is used for proving \lemref{num_MRSs}.


\begin{lem}
  \label{lem:MRS_exist_not_v}
  For any finite $k\in\bbZ_+$ and graph $G=(V,E)$,
  let $S \in \MS_{k}$ be a solution such that $|S|\ge2$. 
  For each $v\in S$,
  there exists a \MinRemSet of $S$ that does not contain $v$.
\end{lem}
\begin{proof}
  The subset $S \setminus \{v\}$ is an \RemSet of $S$
  since $\{v\} \in \MS_{k}$.
  There exists a \MinRemSet $Y$ of $S$
  such that $Y\subseteq S\setminus\{v\}$. 
\end{proof}

\begin{lem}
  \label{lem:num_MRSs}
  For any finite $k\in\bbZ_+$ and graph $G=(V,E)$,
  let $S \in \MS_{k}$ be a solution. 
  \begin{description}
  \item[\rm (i)]
  If $|S| = 1$, then $\MRS_{k}(S)$ is empty.
  \item[\rm (ii)]
  If $|S| \ge 2$, then $|\MRS_{k}(S)| \ge 2$ holds.
  \end{description}
\end{lem}
\begin{proof}
  (i) is obvious.
  For (ii), let $v\in S$ be any vertex.
  By \lemref{MRS_exist_not_v}, there exists a \MinRemSet of $S$, say $Y$,
  such that $v\notin Y$. For $v'\in Y$, 
  there exists a \MinRemSet of $S$, say $Y'$, such that $v'\notin Y'$. 
  Apparently, $Y \neq Y'$ holds.
\end{proof}

\section{Linear-Time Generation of \MaxProSubSols in Strongly-Connected Systems}
\label{sec:strong}
Assume that a strongly-connected graph $G=(V,E)$ is given such that $n = |V| \geq 2$.
In this section, we present a linear-time algorithm to generate all $\MS_{G,1}$-\MaxProSubSols of $V$.

For simplicity, we may employ abuse notations in this section such that
an $\MS_{G,1}$-\MaxProSubSol of $V$ is represented as a ``\MaxProSubSol of $G$'', 
an $\MS_{G^{\top},1}$-\MinRemSet of $V$ is represented as a ``\MinRemSet of $G^{\top}$'',
and so on. 

A naive approach to generating \MaxProSubSols of $G$ is 
that we generate strongly-connected components of
$G-v$ for all $v \in V$,
and then remove duplication and non-maximal solutions somehow.
The time complexity of this method is
$\Omega(n(n+m) + N)$, where $N$ denotes the sum of $|X|$ for all $X\in\MPSS_\MS(V)$. 
Note that $N$ can be $\Omega(n^2)$;
e.g., a complete digraph.


Our purpose is to generate \MaxProSubSols of $G$ in linear time with respect to $n$, $m$ and $N$.
Recall that \MaxProSubSol is the complement of \MinRemSet;
if $Y$ is a \MinRemSet of $G$, then $V\setminus Y$ is a \MaxProSubSol of $G$.
We derive several properties of  \MinRemSets of $G$ in \secref{prop_MRS}. 
In particular, 
for any vertex $s\in V$,
we show that every \MinRemSet $Y$ of $G$ with $s\notin Y$
exists as a path
in the dominator trees $\DT_s$ and $\DT^{\top}_s$
such that the end vertex is a leaf. 
This enables us to generate all \MinRemSets $Y$ with $s\notin Y$
by traversing $\DT_s$ and $\DT^{\top}_s$. 
By \thmref{partition}, $G$ is \MaxProSubSol-disjoint and/or \MinRemSet-disjoint.
In \secref{detmine_MRS_disjointness},
we show how to decide whether $G$ is \MinRemSet-disjoint or not
and explain algorithms to generate all \MaxProSubSols of $G$ in each case. 
This section concludes with the proof for \thmref{strong}.

\subsection{Properties of \MinRemSets}
\label{sec:prop_MRS}

Let us confirm that we can enlarge 
any solution $S\in\MS_1$ by ``attaching'' a directed path to $S$. 

\begin{lem}
  \label{lem:str_path_adding}
  Let $G=(V,E)$ be a digraph and $S \in \MS_1$ be a solution. 
  For any directed path $P = \{v_1, v_2, \dots, v_q\}$ in $G - S$,
  if $v_1 \in \outNeigh{G}{S}$ and $v_q \in \inNeigh{G}{S}$,
  then $S \cup P \in \MS_1$.
\end{lem}

\lemref{str_any_MRS_path} states that, for any \MinRemSet $Y$ of a solution $S$, 
there is a Hamiltonian path in $G[Y]$ for which no ``shortcut'' exists. 
It is also shown that the starting vertex of the Hamiltonian path 
is the only vertex in $Y$ that has incoming edges from $X:=S\setminus Y$;
and that the end vertex is the only vertex in $Y$
that has outgoing edges to $X$.

\begin{lem}
  \label{lem:str_any_MRS_path}
  For a strongly-connected graph $G=(V,E)$,
  let $Y\in\MRS_1(V)$ and $X:=V\setminus Y\in\MPSS_1(V)$. 
  \begin{description}
  \item[\rm (i)]
    In $G[Y]$,
    a path exists
    from a vertex in $\outNeigh{G}{X}$
    to a vertex in $\inNeigh{G}{X}$
    and such a path is Hamiltonian. 
  \item[\rm (ii)] 
    Let us denote the Hamiltonian path in $G[Y]$
    by $v_1\to v_2\to\dots\to v_q$.
    For $i, j \in [1, q]$, if $j - i \geq 2$, then $(v_i, v_j) \notin E$.
  \item[\rm (iii)] It holds that $|\outNeigh{G}{X}| = |\inNeigh{G}{X}| = 1$.
  \end{description}
\end{lem}
\begin{proof}
  (i)
  Observe that neither $\outNeigh{G}{X}$ nor $\inNeigh{G}{X}$ is empty
  since $G$ is strongly-connected. 
  Let $P := \{u_1, u_2, \dots, u_{\ell}\} \subseteq Y$ be a simple path
  such that $u_1\in\outNeigh{G}{X}$ and $u_\ell\in \inNeigh{G}{X}$.
  Such $P$ exists since there is a directed path in $G$
  from $u_1$ to any vertex in $X$ by the strong-connectivity of $G$. 
  If $P\subsetneq Y$, then by \lemref{str_path_adding},
  $G[(V \setminus Y) \cup P] = G[V \setminus (Y \setminus P)]$ is strongly-connected.
  The nonempty subset $Y \setminus P$ would be an \RemSet of $V$,
  which contradicts the minimality of $Y$.
  We see that $P=Y$ holds and it is a Hamiltonian path in $G[Y]$.

  \noindent
  (ii) Suppose $(v_i,v_j)\in E$. 
  Then $Y' := \{v_1, \dots, v_i, v_j, \dots, v_q\}\subsetneq Y$
  would be a simple non-Hamiltonian path in $G[Y]$ such that $v_1\in\outNeigh{G}{X}$ and $v_q\in \inNeigh{G}{X}$,
  contradicting (i).
  
  \noindent
  (iii) If $|\outNeigh{G}{X}| > 1$, 
  there is $i \in [2, q]$ such that $v_i \in \outNeigh{G}{X}$.
  Then $\{v_i,v_{i+1},\dots,v_q\}\subsetneq Y$
  would be a simple non-Hamiltonian path in $G[Y]$ such that $v_i\in\outNeigh{G}{X}$ and $v_q\in \inNeigh{G}{X}$,
  contradicting (i).
  We can show $|\inNeigh{G}{X}| =1$ in an analogous way.
\end{proof}

\paragraph{\MinRemSets and dominator trees.}
For a strongly-connected digraph $G=(V,E)$ and $s\in V$,
we consider where \MinRemSets of $G$ exist in the dominator tree $\DT_s$. 
Let us confirm fundamental properties on transpose and dominator tree. 

\begin{lem}
  \label{lem:str_RS_in_DT}
  For a strongly-connected graph $G=(V,E)$,
  let $Y\subseteq V$, $X:=V\setminus Y$ and $s,t\in X$ $(s\ne t)$.
  \begin{description}
  \item[\rm (i)] $Y$ is an \RemSet of $G$ if and only if $Y$ is an \RemSet of $G^{\top}$. 
  \item[\rm (ii)] $Y$ is a \MinRemSet of $G$ if and only if $Y$ is a \MinRemSet of $G^{\top}$. 
  \item[\rm (iii)] If $Y$ is an \RemSet of $G$, then
    $t$ is reachable from $s$ in both $\DT_s-Y$ and $\DT^{\top}_s-Y$. 
  \end{description}
\end{lem}
\begin{proof}
  (i) and (ii) are immediate.
  For (iii), there is an $s,t$-path in the subgraph $G-Y$ since $Y$ is an \RemSet of $G$,
  indicating that no vertex in $Y$ dominates $t$ in $G_s$.
  Then in $\DT_s$, no vertex in $Y$ exists on the $s,t$-path.
  It is also the same for $\DT^{\top}_s-Y$. 
\end{proof}


\noindent
The converse of \lemref{str_RS_in_DT}(iii) does not hold;
a counterexample is shown in \figref{conv_lem_RS}.

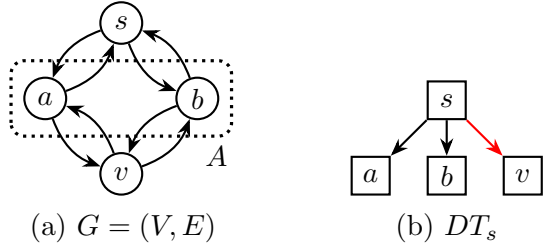
\begin{figure}[t!]
  \centering
  \begin{minipage}[b]{0.3\hsize}
    \centering
    \begin{tikzpicture}[entity-vertex/.style={draw, circle,minimum width=0.55cm, ,inner sep=-1mm, thick}]
        \node[entity-vertex] (vx) at (0, 1) {$s$};
        \node[entity-vertex] (vy) at (0, -1) {$v$};
        \node[entity-vertex] (va) at (-1, 0) {$a$};
        \node[entity-vertex] (vb) at (1, 0) {$b$};
        \draw[thick,-{Stealth}] (vx)to[out=200,in=70](va);
        \draw[thick,-{Stealth}] (va)to[out=-70,in=160](vy);
        \draw[thick,-{Stealth}] (vy)to[out=20,in=-110](vb);
        \draw[thick,-{Stealth}] (vb)to[out=110,in=-20](vx);
        
        \draw[thick,-{Stealth}] (va)to[out=20,in=-110](vx);
        \draw[thick,-{Stealth}] (vx)to[out=-70,in=160](vb);
        \draw[thick,-{Stealth}] (vb)to[out=-160,in=70](vy);
        \draw[thick,-{Stealth}] (vy)to[out=110,in=-20](va);

        \draw[rounded corners=5,very thick,dotted]($ (va) + (-0.45,-0.5) $)--($ (va) + (-0.45,0.5) $)--($ (vb) + (0.45,0.5) $)--($ (vb) + (0.45,-0.5) $)--cycle;
        \draw($ (vb) + (0.25,-0.5) $)node[below]{$ A $};

    \end{tikzpicture}
    \\(a) $G = (V, E)$
\end{minipage}
\begin{minipage}[b]{0.25\hsize}
    \centering
    \begin{tikzpicture}[entity-vertex/.style={draw,minimum width=0.5cm, minimum height=0.5cm,inner sep=-1mm, thick}]
        \node[entity-vertex] (vx) at (0, 1) {$s$};
        \node[entity-vertex] (va) at (-1, 0) {$a$};
        \node[entity-vertex] (vb) at (0, 0) {$b$};
        \node[entity-vertex] (vy) at (1, 0) {$v$};

        \draw[thick,-{Stealth}] (vx)--(va);
        \draw[thick,-{Stealth}] (vx)--(vb);
        \draw[thick,-{Stealth},red] (vx)--(vy);

    \end{tikzpicture}
    \\ (b) $DT_s$
\end{minipage}


  \caption{A counterexample of the converse of \lemref{str_RS_in_DT}(iii).
  (a) A strongly-connected graph~$G$.
  (b) The dominator tree $\DT_s$ of $G_s$.
  Although every vertex of $V \setminus A = \{s, v\}$ is reachable from $s$ in (b),
  $G - A$ is not strongly-connected.}
  \label{fig:conv_lem_RS}
\end{figure}

\lemref{str_MRS_in_DT} provides a necessary and sufficient condition
that a subset $Y\subseteq V$ with $s\notin Y$ is a \MinRemSet of $G$.

\begin{lem}
  \label{lem:str_MRS_in_DT}
  For a strongly-connected graph $G=(V,E)$ and $s\in V$,
  let $Y \subsetneq V$ be a nonempty proper subset such that $s\notin Y$.  
  Then $Y$ is a \MinRemSet of $G$
  if and only if the following \textup{(i)} to \textup{(iii)} hold. 
  \begin{description}
  \item[\rm (i)] For each $u\in Y$,
    it holds that $\Ch_s(u)\subseteq Y$ and $\Ch^{\top}_s(u)\subseteq Y$. 
  \item[\rm (ii)] $\DT_s[Y]$ is a directed path.
  \item[\rm (iii)] $\DT^{\top}_s[Y]$ is a directed path whose orientation is the reverse of $\DT_s[Y]$.
  \end{description}
\end{lem}
\begin{proof}
  Let  $Y = \{v_1, v_2, \dots, v_q\}$.

  \noindent
  $(\Longrightarrow)$
  By \lemref{str_any_MRS_path}(i),
  suppose that $v_1\to v_2\to\dots\to v_q$ is the Hamiltonian path in $G[Y]$.
  
  \noindent (i) If not, then for some $u\in Y$, there would be a child $t\in\Ch_s(u)$ (or $t\in\Ch^{\top}_s(u)$) 
  that cannot be reached from $s$ in $\DT_s-Y$ (or $\DT^{\top}_s-Y$),
  contradicting \lemref{str_RS_in_DT}(iii).
  
  \noindent
  (ii)
  In $G_s$, 
  every $s,v_i$-path, $i\in[1,q]$ visits $v_1$ by \lemref{str_any_MRS_path}(iii),
  and by \lemref{str_any_MRS_path}(ii),
  the only $v_1,v_i$-path is
  $v_1 \rightarrow v_2 \rightarrow \dots \rightarrow v_i$. 
  We have that $\dom(v_i) =\dom(v_{i-1}) \cup \{v_i\}$ for any $i \in [2, q]$,
  meaning that
  $v_1 \rightarrow v_2 \rightarrow \dots \rightarrow v_q$
  is a directed path in $\DT_s$.
  
  \noindent
  (iii) $G^{\top}$ is strongly-connected and $Y$ is a \MinRemSet of $G^{\top}$.
  By \lemref{str_any_MRS_path}(i), there is a Hamiltonian path in $G^{\top}[Y]$,
  where we see that it is $v_q\to v_{q-1}\to\dots\to v_1$. 
  Applying the same argument as (ii) to $G^{\top}_s$,
  $v_q\to v_{q-1}\to\dots\to v_1$ is a directed path in $\DT^{\top}_s$. 

  \medskip
  
  \noindent
  $(\Longleftarrow)$
  By (ii), we represent $\DT_s[Y]$ as a directed path
  $v_1 \rightarrow v_2 \rightarrow \dots \rightarrow v_q$,
  where $\DT^{\top}_s[Y]$ is represented as
  $v_q \rightarrow v_{q-1} \rightarrow \dots \rightarrow v_1$ by (iii).
  By (i), every vertex $t\in V\setminus Y$ is reachable from $s$
  in both $\DT_s-Y$ and $\DT^{\top}_s-Y$,
  indicating that there are $s,t$-path and $t,s$-path in $G-Y$.
  We see that $G-Y$ is strongly-connected by transitivity of strong-connectivity 
  and that $Y$ is an \RemSet of $G$.

  \noindent
  We show the minimality of $Y$. It is trivial if $|Y|=1$. Let $|Y|\ge2$. 
  Suppose any \MinRemSet $Y'$ of $G$ such that $Y'\subseteq Y$.  
  If $v_1\in Y\setminus Y'$, then
  any vertex in $Y'$ would dominate $v_1$ in $G^{\top}_s$
  and $v_1$ is not reachable from $s$ in $\DT^{\top}_s-Y'$,
  which contradicts \lemref{str_RS_in_DT}(iii).
  Hence we have $v_1\in Y'$. 
  If there is $i \in [2, q]$ such that $v_i \in Y\setminus Y'$, then
  $v_1$ would dominate $v_i$ in $G_s$,
  and $v_i$ is not reachable from $s$ in $\DT_s-Y'$,
  which contradicts \lemref{str_RS_in_DT}(iii).
  Thus, $Y' = Y$ holds.
\end{proof}


By (i) and (ii) of \lemref{str_MRS_in_DT},
any \MinRemSet $Y$ of $G$ with $s\notin Y$ contains a leaf in $\DT_s$.
Note that there may be a leaf of $\DT_s$ that does not belong to any \MinRemSet of $G$;
see \figref{leaf_not_in_MRS} for example.

\begin{figure}[t!]
  \centering
  \begin{minipage}[t]{0.225\linewidth}
    \centering
    \begin{tikzpicture}[entity-vertex/.style={draw, circle,minimum width=0.55cm, ,inner sep=-1mm, thick}]
        \node[entity-vertex] (v1) at ($(90:1.25)$) {$v_1$};
        \node[entity-vertex] (v2) at ($(90+360/7:1.25)$) {$v_2$};
        \node[entity-vertex] (v3) at ($(90+720/7:1.25)$) {$v_3$};
        \node[entity-vertex] (v4) at ($(90+3*360/7:1.25)$) {$v_4$};
        \node[entity-vertex] (v5) at ($(90-360/7:1.25)$) {$v_5$};
        \node[entity-vertex] (v6) at ($(90-720/7:1.25)$) {$v_6$};
        \node[entity-vertex] (v7) at ($(90-3*360/7:1.25)$) {$v_7$};
        \node[entity-vertex] (v8) at (0, 0) {$v_8$};
        \node[entity-vertex] (s) at ($(v1)+(0, 1.25)$) {$s$};
        \draw[thick,-{Stealth}] (v1)--(v2);
        \draw[thick,-{Stealth}] (v2)--(v3);
        \draw[thick,-{Stealth}] (v3)--(v4);
        \draw[thick,-{Stealth}] (v4)--(v2);
        \draw[thick,-{Stealth}] (v4)--(v7);
        \draw[thick,-{Stealth}] (v1)--(v5);
        \draw[thick,-{Stealth}] (v5)--(v6);
        \draw[thick,-{Stealth}] (v6)--(v7);
        \draw[thick,-{Stealth}] (v7)--(v8);
        \draw[thick,-{Stealth}] (v8)--(v1);
        \draw[thick,-{Stealth}] (s)to[out=-135,in=135](v1);
        \draw[thick,-{Stealth}] (v1)to[out=45,in=-45](s);

    \end{tikzpicture}
    $G = (V, E)$
\end{minipage}
\hfill
\begin{minipage}[t]{0.225\linewidth}
    \centering
    \begin{tikzpicture}[entity-vertex/.style={draw, circle,minimum width=0.55cm, ,inner sep=-1mm, thick}]

        \node[entity-vertex] (v1) at ($(90:1.25)$) {$v_1$};
        \node[entity-vertex] (v2) at ($(90+360/7:1.25)$) {$v_2$};
        \node[entity-vertex] (v3) at ($(90+720/7:1.25)$) {$v_3$};
        \node[entity-vertex] (v4) at ($(90+3*360/7:1.25)$) {$v_4$};
        \node[entity-vertex] (v5) at ($(90-360/7:1.25)$) {$v_5$};
        \node[entity-vertex] (v6) at ($(90-720/7:1.25)$) {$v_6$};
        \node[entity-vertex] (v7) at ($(90-3*360/7:1.25)$) {$v_7$};
        \node[entity-vertex] (v8) at (0, 0) {$v_8$};
        \node[entity-vertex] (s) at ($(v1)+(0, 1.25)$) {$s$};
        \draw[thick,{Stealth}-] (v1)--(v2);
        \draw[thick,{Stealth}-] (v2)--(v3);
        \draw[thick,{Stealth}-] (v3)--(v4);
        \draw[thick,{Stealth}-] (v4)--(v2);
        \draw[thick,{Stealth}-] (v4)--(v7);
        \draw[thick,{Stealth}-] (v1)--(v5);
        \draw[thick,{Stealth}-] (v5)--(v6);
        \draw[thick,{Stealth}-] (v6)--(v7);
        \draw[thick,{Stealth}-] (v7)--(v8);
        \draw[thick,{Stealth}-] (v8)--(v1);
        \draw[thick,{Stealth}-] (s)to[out=-135,in=135](v1);
        \draw[thick,{Stealth}-] (v1)to[out=45,in=-45](s);

    \end{tikzpicture}
    $\Trans{G} = (V, \Trans{E})$
\end{minipage}
\hfill
\begin{minipage}[t]{0.225\linewidth}
    \centering
    \begin{tikzpicture}[entity-vertex/.style={draw,minimum width=0.5cm, minimum height=0.5cm,inner sep=-1mm, thick}]
        \node[entity-vertex] (s) at (0, 1) {$s$};
        \node[entity-vertex] (v1) at (0, 0) {$v_1$};
        \node[entity-vertex] (v2) at (-1, -1) {$v_2$};
        \node[entity-vertex] (v3) at (-1, -2) {$v_3$};
        \node[entity-vertex] (v4) at (-1, -3) {$v_4$};
        \node[entity-vertex] (v5) at (0, -1) {$v_5$};
        \node[entity-vertex] (v6) at (0, -2) {$v_6$};
        \node[entity-vertex] (v7) at (1, -1) {$v_7$};
        \node[entity-vertex] (v8) at (1, -2) {$v_8$};

        \draw[thick,-{Stealth}] (s)--(v1);
        \draw[thick,-{Stealth}] (v1)--(v2);
        \draw[thick,-{Stealth}] (v2)--(v3);
        \draw[thick,-{Stealth}] (v3)--(v4);
        \draw[thick,-{Stealth}] (v1)--(v5);
        \draw[thick,-{Stealth}] (v5)--(v6);
        \draw[thick,-{Stealth}] (v1)--(v7);
        \draw[thick,-{Stealth}] (v7)--(v8);

    \end{tikzpicture}
    
    $\DT_s$
\end{minipage}
\hfill
\begin{minipage}[t]{0.2\linewidth}
    \centering
    \begin{tikzpicture}[entity-vertex/.style={draw,minimum width=0.5cm, minimum height=0.5cm,inner sep=-1mm, thick}]
        \node[entity-vertex] (s) at (0, 1) {$s$};
        \node[entity-vertex] (v1) at (0, 0) {$v_1$};
        \node[entity-vertex] (v8) at (0, -1) {$v_8$};
        \node[entity-vertex] (v7) at (0, -2) {$v_7$};
        \node[entity-vertex] (v4) at (-0.75, -3) {$v_4$};
        \node[entity-vertex] (v3) at (-0.75, -4) {$v_3$};
        \node[entity-vertex] (v2) at (-0.75, -5) {$v_2$};
        \node[entity-vertex] (v6) at (0.75, -3) {$v_6$};
        \node[entity-vertex] (v5) at (0.75, -4) {$v_5$};

        \draw[thick,-{Stealth}] (s)--(v1);
        \draw[thick,-{Stealth}] (v1)--(v8);
        \draw[thick,-{Stealth}] (v8)--(v7);
        \draw[thick,-{Stealth}] (v7)--(v4);
        \draw[thick,-{Stealth}] (v4)--(v3);
        \draw[thick,-{Stealth}] (v3)--(v2);
        \draw[thick,-{Stealth}] (v7)--(v6);
        \draw[thick,-{Stealth}] (v6)--(v5);

    \end{tikzpicture}
    $\Trans{\DT}_s$
\end{minipage}
  \caption{An example where a leaf of the dominator tree does not belong to a \MinRemSet.
  The graph $G$ is strongly-connected and 
  $\MRS_{\MS_1}(V) = \{ \{s\}, \{v_2, v_3, v_4\}, \{v_5, v_6\} \}$. 
  Although vertex~$v_8$ is a leaf of $\DT_s$,
  it does not belong to any \MinRemSet of $V$.}
  \label{fig:leaf_not_in_MRS}
\end{figure}
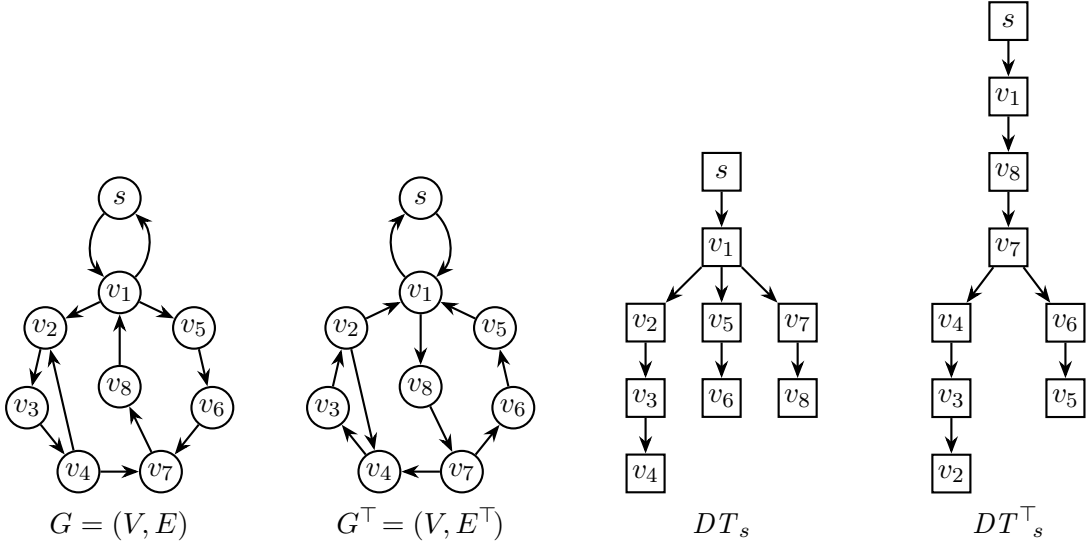

Now we are ready to describe an algorithm to generate all \MinRemSets of $G$
that does not contain the root $s$. The algorithm is shown in \algref{MRS_single_exclude}. 

\begin{algorithm}[t]
  \caption{An algorithm to generate all \MinRemSets of $V$ that do not contain a given vertex for a strongly-connected graph $G = (V, E)$}
  \label{alg:MRS_single_exclude}
  \begin{algorithmic}[1]
    \Require A strongly-connected graph $G=(V,E)$ such that $|V|\ge2$
    and a vertex $s \in V$
    \Ensure All \MinRemSets of $V$ that do not contain $s$
    \Procedure{\GENMRSWOUTRT}{$G = (V, E), s$}
    \State Construct the dominator trees $\DT_s$ and $\DT^{\top}_s$;
    \For{\textbf{each} leaf $v$ of $\DT_s$ such that $|\Ch^{\top}_s(v)| \leq 1$}
      \State $u := v$; $Y := \{u\}$;
      \While{$u$ is not a leaf of $\DT^{\top}_s$,
        $|\Ch_s(\pi(u))| = 1$, 
        $|\Ch^{\top}_s(\pi(u))| \leq 1$ and
        $\pi^{\top}(\pi(u)) = u$
      }\label{codeline:while_1}
        \State $u:= \pi(u)$; $Y := Y \cup \{u\}$
      \EndWhile;\label{codeline:while_1_end}
      \If{$u$ is a leaf of $\DT^{\top}_s$}\label{codeline:while_1_if}
        \State Output $Y$
      \EndIf
    \EndFor
    \EndProcedure
  \end{algorithmic}
\end{algorithm}

\begin{lem}
  \label{lem:str_MRS_single_exclude}
  For a strongly-connected graph $G=(V,E)$ such that $|V|\ge2$ and $s \in V$,
  Algorithm~{\rm\ref{alg:MRS_single_exclude}} generates all \MinRemSets $Y$\/of $G$
  such that $s \notin Y$ in $O(n + m)$ time and space.
\end{lem}
\begin{proof}
  (Correctness) The while-loop from line~\ref{codeline:while_1} to line~\ref{codeline:while_1_end}
  begins with $u:=v$ for a leaf $v$ of $\DT_s$
  and repeats $u:=\pi(u)$ (i.e., we move to the parent of $u$) and $Y:=Y\cup\{u\}$
  as long as the conditions in \lineref{while_1} are satisfied.
  \begin{itemize}
  \item It holds that $s\notin Y$ for any output $Y$;
    If $u$ is a child of $s$,
    then the last condition $\pi^{\top}(\pi(u))=u$ in \lineref{while_1}
    is not satisfied
    since $\pi^{\top}(\pi(u))=\pi^{\top}(s)=\NIL\ne u$.
  \item We claim that each leaf in $\DT_s$ is contained in at most one \MinRemSet $Z$ such that $s\notin Z$. 
    By \lemref{str_MRS_in_DT}(ii), $\DT_s[Z]$ is a directed path. 
    Also by (i), it is a path from a vertex $a$ to a leaf $b$, where it is possible that $a=b$.
    If there are two \MinRemSets that contain $b$,
    it would contradict that each vertex except the root has its unique parent.     
  \item We show that, for each leaf $v$ of $\DT_s$,
    the algorithm outputs the \MinRemSet $Z$ such that $s\notin Z$ and $v\in Z$
    if exists. 
    When we exit the while-loop, we already observed that $s\notin Y$ holds.
    $Y$ is the set of vertices between a leaf $v$ and its ancestor $u$ in $\DT_s$
    and hence $\DT_s[Y]$ is a directed path and satisfies \lemref{str_MRS_in_DT}(ii). 
    We also see that $\DT^{\top}_s[Y]$ is a directed path
    and hence satisfies \lemref{str_MRS_in_DT}(iii);
    by the last condition of \lineref{while_1},
    $Y$ is the set of vertices between an inner node $v$ and its descendant $u$
    in $\DT^{\top}_s$. 
    For \lemref{str_MRS_in_DT}(i), $|\Ch_s(a)|\le 1$ and $|\Ch^{\top}_s(a)|\le1$ hold for every $a\in Y$. 
    While $\Ch_s(a)\subseteq Y$ is satisfied for every $a\in Y$,
    $\Ch^{\top}_s(u)\subseteq Y$ is not satisfied if $u$,
    the end vertex of the directed path $\DT^{\top}_s[Y]$,
    is not a leaf of $\DT^{\top}_s$.
    However, such $Y$ is not output by \lineref{while_1_if}.
    We have shown that the output $Y$ satisfies \lemref{str_MRS_in_DT}(i) to (iii)
    and that $Y=Z$ holds.  
  \end{itemize}
  
  \noindent
  (Time Complexity Analyses)
  Adam et al.~\cite{Adam1998} show that
  a dominator tree for a flow graph can be constructed in $O(n + m)$ time.
  Each edge in $\DT_s$ and $\DT^{\top}_s$ is scanned at most once.

  \noindent
  (Space Complexity Analyses)
  A dominator tree for a flow graph can be constructed and stored in $O(n + m)$ space. 
  The size of $Y$ is $O(n)$.
  We can store $\pi(v)$, $\Ch(v)$, $\pi^{\top}(v)$ and $\Ch^{\top}(v)$ for all $v \in V$
  in $O(n)$ space.
\end{proof}

\algref{MRS_single_exclude} outputs 
all \MinRemSets of $V$ that do not contain the designated vertex $s\in V$. 
Let us observe that, using the dominator tree 
constructed in the algorithm, 
we can output a \MaxProSubSol $X:=V\setminus Y$ for each $Y\in\MRS(V)$ with $s\notin Y$
in $O(|X|)$ time. 

\begin{lem}
  \label{lem:convert_MRS_to_MPSS}
  For a strongly-connected graph $G=(V,E)$ and $s\in V$,
  let $\MY$ denote the family of \MinRemSets
  that Algorithm~{\rm\ref{alg:MRS_single_exclude}}
  outputs for the input $(G,s)$. 
  For each $Y\in\MY$, 
  we can output $X := V \setminus Y$ in $O(|X|)$ time and $O(n)$ space.
\end{lem}
\begin{proof}
  Let $\MY=\{Y_1,Y_2,\dots,Y_q\}$.
  When we construct each $Y_i\in\MY$,
  $i\in[1,q]$ during \algref{MRS_single_exclude},
  let us mark each vertex in $Y_i$ with the integer $i$. 
  Recall that the sets in $\MY$ are pairwise disjoint
  and that $Y_i$ induces a directed path in $\DT_s$
  between a leaf $v_i$ and one of its ancestors, say $a_i$.
  This does not change complexity bounds in \lemref{str_MRS_single_exclude}. 

  Suppose that we maintain the dominator tree $\DT_s$
  after \algref{MRS_single_exclude} halts. 
  To output $X_i:=V\setminus Y_i$,
  it suffices to traverse $\DT_s$ from the root $s$
  by a conventional graph search algorithm (e.g., DFS/BFS),
  skipping the vertices marked by $i$, that is $Y_i$. 
  We can output $X_i=V\setminus Y_i$ in $O(|X|)$ time in this way
  and store $\DT_s$ in $O(n)$ space.
\end{proof}

\subsection{Generating \MaxProSubSols}
\label{sec:detmine_MRS_disjointness}

How we generate \MaxProSubSols of $G$
depends on whether $G$ is \MinRemSet-disjoint or not.
We can decide whether $G$ is \MinRemSet-disjoint or not as follows.
Choosing a vertex $s\in V$,
we first generate all \MinRemSets $Y\in\MRS_1(V)$ such that $s\notin Y$
by \algref{MRS_single_exclude}. 
The \MinRemSets generated by \algref{MRS_single_exclude} are pairwise
disjoint. If two or more \MinRemSets are generated,
then $G$ is \MinRemSet-disjoint. Otherwise, pick up a vertex
in the generated \MinRemSet, say $t$.
We again run \algref{MRS_single_exclude} to 
generate all \MinRemSets $Y\in\MRS_1(V)$ such that $t\notin Y$.
If there are two or more \MinRemSets are generated,
then $G$ is \MinRemSet-disjoint. Otherwise,
if the two \MinRemSets generated by the two runs are disjoint,
then $G$ is \MinRemSet-disjoint. 

\algref{MRS_determine_disjointness} summarizes this algorithm. 
The following lemma shows the correctness and
the complexity of \algref{MRS_determine_disjointness}. 


\begin{algorithm}[t]
  \caption{An algorithm to determine the \MinRemSet-disjointness of a given strongly-connected graph $G = (V, E)$}
  \label{alg:MRS_determine_disjointness}
  \begin{algorithmic}[1]
    \Require A strongly-connected graph $G=(V,E)$  such that $|V|\ge2$
    \Ensure \True if $G$ is \MinRemSet-disjoint; \False otherwise
    \State $\textsc{is-\MinRemSet-disjoint} := \False$;
    \State Let $s \in V$ be any vertex;\label{codeline:s_chosen}
    \State $\MY_s := $ the family of all outputs of \Call{\GENMRSWOUTRT}{$G, s$};
     \If{$|\MY_s| >1$}
        \State $\textsc{is-\MinRemSet-disjoint} := \True$
    \Else
        \State Let $t \in Y_s$ be any vertex in the unique set $Y_s \in \MY_s$;\label{codeline:t_chosen}
        \State $\MY_t := $ the family of all outputs of \Call{\GENMRSWOUTRT}{$G, t$};
        \If{$|\MY_t| >1$}
            \State $\textsc{is-\MinRemSet-disjoint} := \True$
        \Else
            \State $Y_t := $ the unique set in $\MY_t$;
            \If{$Y_s \cap Y_t = \emptyset$}
                \State $\textsc{is-\MinRemSet-disjoint} := \True$
            \EndIf
        \EndIf
    \EndIf;
    \State Output $\textsc{is-\MinRemSet-disjoint}$
  \end{algorithmic}
\end{algorithm}

\begin{lem}
  \label{lem:str_MRS_determine_disjoint}
  For a strongly-connected graph $G=(V,E)$ such that $|V|\ge2$,
  \algref{MRS_determine_disjointness} identifies
  whether $G$ is \MinRemSet-disjoint or not
  in $O(n + m)$ time and space.
\end{lem}
\begin{proof}
  For $v\in V$,
  there is a \MinRemSet of $G$ that does not contain $v$ by \lemref{MRS_exist_not_v}.
  The \MinRemSets output by \algref{MRS_single_exclude} are disjoint,
  and we can conclude that
  $G$ is \MinRemSet-disjoint by \lemref{SSD_basic_with_Y}(ii) 
  if $|\MY_s|>1$ or $|\MY_t|>1$. 
  Otherwise (i.e., if $|\MY_s| = |\MY_t| = 1$),
  let $\MY_s=\{Y_s\}$ and $\MY_t=\{Y_t\}$. 
  We see that $Y_s \neq Y_t$ holds since $t \in Y_s$ and $t \notin Y_t$.
  Again, by \lemref{SSD_basic_with_Y}(ii),
  we can conclude that $G$ is \MinRemSet-disjoint
  if $Y_s\cap Y_t=\emptyset$.
  Otherwise, $G$ is not \MinRemSet-disjoint
  and hence \MaxProSubSol-disjoint by \thmref{partition}. 
  It is easy to see that the time and space complexities are $O(n+m)$.  
\end{proof}

Using \algref{MRS_determine_disjointness},
we can decide whether $G$ is \MaxProSubSol-disjoint or not as follows.
First, run \algref{MRS_determine_disjointness}. 
If $G$ is not \MinRemSet-disjoint, then
we can conclude that $G$ is \MaxProSubSol-disjoint by \thmref{disjoint}.
Otherwise, it is possible that
$G$ is \MinRemSet-disjoint as well as \MaxProSubSol-disjoint.
This case happens when and only when 
$|\MY_s| = |\MY_t| = 1$ and $Y_s \cup Y_t = V$,
as described in \lemref{sametime}.
The check can be done in $O(n)$ time. 


Now we are ready to present algorithms to generate all \MaxProSubSols of $G$.
In \algref{MRS_disjoint}, we summarize the algorithm
for the case when $G$ is \MinRemSet-disjoint,
and in \algref{MPS_disjoint},
we summarize the algorithm for the case when
$G$ is \MaxProSubSol-disjoint. The correctness
and the complexity analyses are
shown in Lemmas~\ref{lem:str_MRS_disjoint}
and \ref{lem:str_MRS_not_disjoint}, respectively,
followed by the proof for \thmref{strong}.

\begin{algorithm}[t]
  \caption{An algorithm to generate all \MaxProSubSols for a strongly-connected \MinRemSet-disjoint graph $G = (V, E)$}
  \label{alg:MRS_disjoint}
  \begin{algorithmic}[1]
    \Require A strongly-connected \MinRemSet-disjoint graph $G=(V,E)$  such that $|V|\ge2$
    \Ensure All \MaxProSubSols of $G$
    \State Let $s \in V$;
    \State $\MY_s := $ the family of all outputs of \Call{\GENMRSWOUTRT}{$G, s$};
    \For{\textbf{each} $Y_s \in \MY_s$}
      \State Output $V \setminus Y_s$\label{codeline:disjoint_out1}
    \EndFor;
    \State Let $Y_s \in \MY_s$ 
    and $t \in Y_s$;\label{codeline:take_t}
    \State $\MY_t := $ the family of all outputs of \Call{\GENMRSWOUTRT}{$G, t$};
    \For{\textbf{each} $Y_t \in \MY_t$}
        \If{$s \in Y_t$}
            \State Output $V \setminus Y_t$\label{codeline:disjoint_out2}
        \EndIf
    \EndFor
  \end{algorithmic}
\end{algorithm}

\begin{lem}
  \label{lem:str_MRS_disjoint}
  For a strongly-connected \MinRemSet-disjoint graph $G=(V,E)$,
  Algorithm~{\rm\ref{alg:MRS_disjoint}} generates all \MaxProSubSols of $G$
  in $O(n + m + N)$ time and $O(n + m)$ space.
\end{lem}
\begin{proof}
  (Correctness)
  By \lemref{num_MRSs}(ii), there are at least two \MinRemSets of $G$.
  By the \MinRemSet-disjointness of $G$,
  $s$ belongs to at most one \MinRemSet of $G$
  and there exist \MinRemSets of $G$ that do not contain $s$.
  This indicates that $\MY_s$,
  the set of all \MinRemSets of $G$ that do not contain $s$,
  is non-empty.
  For each $Y_s\in\MY_s$, the \MaxProSubSol $V\setminus Y_s$ is  output in \lineref{disjoint_out1}. 
  For any $Y_s\in\MY_s$, let $t\in Y_s$.
  Let $Y$ be a \MinRemSet $Y$ of $S$ that contains $s$. 
  Then $t\notin Y$ should hold
  since $Y_s\ne Y$ is the only \MinRemSet that contains $t$.
  Such $Y$ is contained in $\MY_t$
  and the corresponding \MaxProSubSol $V\setminus Y$ is 
  output in \lineref{disjoint_out2}. 

  \noindent
  (Complexity analyses)
  We have $|\MY_s|,|\MY_t|\le n$. 
  By \lemref{convert_MRS_to_MPSS},
  we can output a \MaxProSubSol $X:=V \setminus Y$ for a \MinRemSet $Y$
  in $O(|X|)$ time and $O(n)$ space, using the dominator tree
  constructed in \Call{\GENMRSWOUTRT}{} in \algref{MRS_single_exclude}. 
  We can run \Call{\GENMRSWOUTRT}{}
  in $O(n+m)$ time and space by \lemref{str_MRS_single_exclude}.
\end{proof}

  

\begin{algorithm}[t]
  \caption{An algorithm to generate all \MaxProSubSols for a strongly-connected
  \MaxProSubSol-disjoint graph $G = (V, E)$ }
  \label{alg:MPS_disjoint}
  \begin{algorithmic}[1]
    \Require  A strongly-connected \MaxProSubSol-disjoint graph $G=(V,E)$  such that $|V|\ge2$
    \Ensure All \MaxProSubSols of $G$
    \State Let $s \in V$;
    \State $Y := $ the unique output of \Call{\GENMRSWOUTRT}{$G, s$};
    \label{codeline:unique}
    \State Output $V \setminus Y$;
    \label{codeline:unique2}
    \State Output all strongly-connected components of $G[Y]$
    \label{codeline:unique3}
  \end{algorithmic}
\end{algorithm}

\begin{lem}
  \label{lem:str_MRS_not_disjoint}
  For a strongly-connected \MaxProSubSol-disjoint graph $G=(V,E)$,
  Algorithm~{\rm\ref{alg:MPS_disjoint}}
  generates all \MaxProSubSols of $G$ in $O(n + m)$ time and space.
\end{lem}
\begin{proof}
  (Correctness)
  Let $s\in V$.
  By \thmref{partition}, $\MPSS_1(V)$ is a partition of $V$, and hence
  there exists precisely one $X\in\MPSS_1(V)$ such that $s\in V$.
  Further, the corresponding \MinRemSet $Y:=V\setminus X$ satisfies $s\notin Y$
  and any other $Y'\in\MRS(V)$ satisfies $s\in Y$. 
  This $Y$ is generated in \lineref{unique}
  and $X=V\setminus Y$ is output in \lineref{unique2}.
  We have $\MPSS_1(V)=\textsc{MaxSS}_1(V)\sqcup\{X\}$ by \lemref{SSD_MRS},
  and all solutions in $\textsc{MaxSS}_1(V)$ are output
  in \lineref{unique3}. 
  
  \noindent
  (Complexity Analyses)
  By \lemref{str_MRS_single_exclude}, $Y$ can be obtained
  in $O(n + m)$ time and space.
  The strongly-connected components of a given digraph can be generated
  in $O(n + m)$ time and space. 
\end{proof}

\paragraph{Proof for \thmref{strong}.}
We can determine whether $G$ is \MinRemSet-disjoint
in $O(n + m)$ time by \lemref{str_MRS_determine_disjoint}.
We can also determine whether $G$ is \MaxProSubSol-disjoint or not
in additional $O(n)$ time computation. 
If $G$ is \MinRemSet-disjoint,
then we can generate \MaxProSubSols of $G$
in $O(n + m + N) = O(n + m + |X_1| + |X_2| + \dots + |X_q|)$ time
(\lemref{str_MRS_disjoint}).
Otherwise, 
we can generate them in $O(n + m)$ time (\lemref{str_MRS_not_disjoint})
since $G$ is \MaxProSubSol-disjoint in this case.
\QED

\section{Linear-Delay Enumeration of Solutions in Strongly-Connected Systems}
\label{sec:linear_delay}

In this section, we propose a linear-delay algorithm
that enumerates all vertex subsets
that induce strongly-connected subgraphs for a given
digraph $G=(V,E)$,
as a proof for \thmref{linear_delay}. 
For this purpose, we consider an enumeration problem
in a more general setting that asks to
enumerate all solutions in a given SSD system.
We show a linear-delay algorithm for this general problem
and then apply the algorithm to
a strongly-connected system $(V,\MS_{G,1})$
which is an SSD system (\lemref{k_edge_SSD}). 

An assumption that the SSD system $(U,\MS)$ is given explicitly
is meaningless in the context of enumeration
since otherwise the problem would be trivial;
it suffices to output $S\in\MS$ one by one.
As is done in previous work (e.g., \cite{AU.2009,BHPW.2010,Conte.2019,TH.2023}), 
we assume that it is given implicitly by means of oracles.
We take up three oracles, $\Lambda_{\textrm{id}}$ and $\Lambda_{\textrm{\MaxProSubSol}}$ and $\Lambda_{\textrm{\MinRemSet}}$,
whose roles are as follows. 
\begin{itemize}
\item $\Lambda_{\textrm{id}}(S)$ (identification): For a solution $S\in\MS$, it returns \True
  if $S$ is \MaxProSubSol-disjoint and \False otherwise. 
\item $\Lambda_{\textrm{\MaxProSubSol}}(S,I)$ (generation of \MaxProSubSols): For any \MaxProSubSol-disjoint solution $S\in\MS$
  and a subset $I\subseteq S$,
  it returns all \MaxProSubSols $X$ of $S$ such that $X\supseteq I$.
\item $\Lambda_{\textrm{\MinRemSet}}(S,I)$ (generation of \MinRemSets): For any \MinRemSet-disjoint solution $S\in\MS$
  and a subset $I\subseteq S$,
  it returns all \MinRemSets $Y$ of $S$ such that $Y\cap I=\emptyset$. 
\end{itemize}

Our strategy is binary partition, a fundamental strategy of enumeration. 
For any set system $(U=\{u_1,u_2,\dots,u_q\},\MS)$ (i.e., not necessarily SSD),
we have a partition of the solution set as follows;
\begin{align}
  \MS=\MS(U,\emptyset)&=\MS(U\setminus\{u_1\},\emptyset)\sqcup\MS(U,\{u_1\})\nonumber\\
  &=\MS(U\setminus\{u_1,u_2\},\emptyset)\sqcup\MS(U\setminus\{u_1\},\{u_2\})\sqcup\MS(U,\{u_1\})\nonumber\\
  &=\MS(\emptyset,\emptyset)\sqcup\big(\bigsqcup_{i=1}^q\MS(U\setminus\{u_1,\dots,u_{i-1}\},\{u_i\})\big).\label{eq:basic_partition}
\end{align}
Based on this partition,
we enumerate solutions in
each $\MS_{G,1}(V\setminus\{v_1,\dots,v_{i-1}\},\{v_i\})$,
$i\in[1,n]$ independently, where $V=\{v_1,v_2,\dots,v_n\}$.
We do not need to take into account the first term
$\MS_{G,1}(\emptyset,\emptyset)$ since it is empty. 
We partition each $\MS_{G,1}(V\setminus\{v_1,\dots,v_{i-1}\},\{v_i\})$
further by using a nice structure that is peculiar to SSD system. 
The following Lemmas~\ref{lem:enum_disunion}
and \ref{lem:enum_partition} provide the partition
of $\MS_{G,1}(V\setminus\{v_1,\dots,v_{i-1}\},\{v_i\})$,
which are
analogous to Lemmas~3 and 4 in \cite{TH.2023}
for SD system;
recall that SSD system is
an extension of SD system.


\begin{lem}
  \label{lem:enum_disunion}
  For an SSD system $(U,\MS)$, let $S\in\MS$ be a solution.
  If $S$ is not \MaxProSubSol-disjoint,
  then it holds that $\MS(S,I)=\MS(S\setminus Y,I)\sqcup\MS(S,I\cup Y)$
  for any nonempty subset $I\subseteq S$
  and $Y\in\MRS_\MS(S,I)$.
\end{lem}
\begin{proof}
  By the definition of $\MRS_\MS(S,I)$ (see \secref{prel_set}),
  $I$ and $Y$ are disjoint.
  The two families $\MS(S\setminus Y,I)$ and $\MS(S,I\cup Y)$ are disjoint
  since every solution in the former contains no element in $Y$,
  whereas every solution in the latter contains $Y$ as a subset.

  For the equality, it is obvious that
  $\MS(S,I)\supseteq \MS(S\setminus Y,I)\sqcup\MS(S,I\cup Y)$ holds.
  To show $\MS(S,I)\subseteq \MS(S\setminus Y,I)\sqcup\MS(S,I\cup Y)$,
  let $S'\in\MS(S,I)$. If $S'=S$, then $S'\in\MS(S,I\cup Y)$ holds.
  If $S'\subsetneq S$, then
  at least one of $Y\supseteq S'$, $Y\subseteq S'$ and $Y\cap S'=\emptyset$
  holds by SSD property.
  Among these, the first case $Y\supseteq S'$ should not happen
  since $Y\supseteq S'\supseteq I$ and hence $Y\cap I=I\ne\emptyset$, a contradiction.
  If $Y\subsetneq S'$, then $S'\in\MS(S,I\cup Y)$ holds. 
  If $Y\cap S'=\emptyset$, then $S'\in\MS(S\setminus Y,I)$ holds. 
\end{proof}

\begin{lem}
  \label{lem:enum_partition}
  For an SSD system $(U,\MS)$, let $S\in\MS$ be a solution.
  If $S$ is not \MaxProSubSol-disjoint, then
  for any nonempty subset $I\subseteq S$, it holds that
  \[
  \MS(S,I)=\{S\}\sqcup\big(\bigsqcup_{i=1}^q\MS(S\setminus Y_i,I\sqcup Y_1\sqcup\dots\sqcup Y_{i-1})\big),
  \]
  where $\MRS_\MS(S,I):=\{Y_1,Y_2,\dots,Y_q\}$. 
\end{lem}
\begin{proof}
  By \lemref{enum_disunion},
  we have a partition $\MS(S,I)=\MS(S\setminus Y_1,I)\sqcup\MS(S,I\sqcup Y_1)$.
  By \thmref{disjoint}, $S$ is \MinRemSet-disjoint
  and hence $Y_1\cap Y_i=\emptyset$, $i\in[2,q]$ and $Y_i\in\MRS_\MS(S,I\sqcup Y_1)$.
  By definition, $\MRS_\MS(S,I\sqcup Y_1)$ is a subset of
  $\MRS_\MS(S,I)=\{Y_1,Y_2,\dots,Y_q\}$, and thus 
  $\MRS_\MS(S,I\sqcup Y_1)=\{Y_2,\dots,Y_q\}$. 
  Applying \lemref{enum_disunion} recursively, it holds that
  \begin{align*}
    \MS(S,I)
    &=\MS(S\setminus Y_1,I)\sqcup\MS(S,I\sqcup Y_1)\\
    &=\MS(S\setminus Y_1,I)\sqcup\MS(S\setminus Y_2,I\sqcup Y_1)\sqcup\MS(S,I\sqcup Y_1\sqcup Y_2)\\
    &=\big(\bigsqcup_{i=1}^q\MS(S\setminus Y_i,I\sqcup Y_1\sqcup\dots\sqcup Y_{i-1})\big)\sqcup\MS(S,I\sqcup Y_1\sqcup\dots\sqcup Y_q),
  \end{align*}
  where $\MS(S,I\sqcup Y_1\sqcup\dots\sqcup Y_q)=\{S\}$ holds
  since no \MinRemSet of $S$ that is disjoint with $I\sqcup Y_1\sqcup\dots\sqcup Y_q$ exists. 
\end{proof}

For an SSD system $(U,\MS)$ and a nonempty subset $I\subseteq U$,
let $S\in\MS$ be a solution such that $S\supseteq I$. 
If $S$ is \MaxProSubSol-disjoint, then
there is an immediate partition of $\MS(S,I)$ such that
\begin{align}
  \MS(S,I)=\{S\}\sqcup\big(\bigsqcup_{X\in\MPSS_\MS(S):\ I\subseteq X}\MS(X,I)\big),
  \label{eq:enum_dis}
\end{align}
where a \MaxProSubSol $X$ of $S$ with $I\subseteq X$ is unique
since $S$ is \MaxProSubSol-disjoint and $I$ is nonempty.
Otherwise (i.e., if $S$ is \MinRemSet-disjoint), 
\lemref{enum_partition} provides a partition of $\MS(S,I)$ such that
\begin{align}
  \MS(S,I)=\{S\}\sqcup\big(\bigsqcup_{i=1}^q\MS(S_i,I_i)\big),
  \label{eq:enum_no_dis}
\end{align}
where $\MRS_\MS(S,I)=\{Y_1,Y_2,\dots,Y_q\}$,
$S_i=S\setminus Y_i$ and $I_i=I\sqcup Y_1\sqcup\dots\sqcup Y_{i-1}$, $i\in[1,q]$.
%
We can enumerate all solutions in $\MS(S,I)$ as follows.
\begin{itemize}
\item If $S$ is \MaxProSubSol-disjoint, then we output solutions in $\MS(X,I)$
  for $X\in\MPSS_\MS(S)$ with $X\supseteq I$ recursively; and
\item otherwise, we output solutions in $\MS(S_i,I_i)$, $i\in[1,q]$
  recursively. 
\end{itemize}
The solution $S$ can be output
before or after the above solutions. 

This procedure is summarized in \algref{linear_delay}.
Let $\tau_{\textrm{id}}$, $\tau_{\textrm{\MaxProSubSol}}$ and $\tau_{\textrm{\MinRemSet}}$
denote upper bounds on the computation time of oracles
$\Lambda_{\textrm{id}}$, $\Lambda_{\textrm{\MaxProSubSol}}$ and $\Lambda_{\textrm{\MinRemSet}}$, respectively.
Let $\sigma_{\textrm{id}}$, $\sigma_{\textrm{\MaxProSubSol}}$ and $\sigma_{\textrm{\MinRemSet}}$
denote upper bounds on the computation space used by oracles
$\Lambda_{\textrm{id}}$, $\Lambda_{\textrm{\MaxProSubSol}}$ and $\Lambda_{\textrm{\MinRemSet}}$, respectively. 

\begin{algorithm}[t]
  \caption{An algorithm to enumerate all solutions in $\MS(S,I)$ for a solution $S$ in an SSD system and its nonempty subset $I\subseteq S$}
  \label{alg:linear_delay}
  \begin{algorithmic}[1]
    \Require An SSD system $(U,\MS)$ that is implicitly given
    by means of oracles $\Lambda_{\textrm{id}}$, $\Lambda_{\textrm{\MaxProSubSol}}$ and $\Lambda_{\textrm{\MinRemSet}}$,
    a solution $S\in\MS$ and a nonempty subset $I\subseteq S$ 
    \Ensure All solutions in $\MS(S,I)$
    \Procedure{EnumSSD}{$S,I,d$}
    \State Output $S$ {\bf if} $d$ is even;
    \If{$\Lambda_{\textrm{id}}(S)$ returns \True}\label{codeline:enum_id}
    \State $\MX:=\Lambda_{\textrm{\MaxProSubSol}}(S,I)$;
    \For{{\bf each} $X\in\MX$}
    \State Execute \Call{EnumSSD}{$X,I,d+1$} \label{codeline:enum_recur_dis}
    \EndFor    
    \Else
    \State $J:=I$; $\MY:=\Lambda_{\textrm{\MinRemSet}}(S,I)$;
    \For{{\bf each} $Y\in\MY$}
    \State Execute \Call{EnumSSD}{$S\setminus Y,J,d+1$}; $J:=J\cup Y$ \label{codeline:enum_recur_no_dis}
    \EndFor
    \EndIf;
    \State Output $S$ {\bf if} $d$ is odd
    \EndProcedure
  \end{algorithmic}
\end{algorithm}

\begin{thm}
  \label{thm:linear_delay_SSD}  
  Let $(U,\MS)$ be an SSD system that is implicitly given
  by oracles $\Lambda_{\textrm{id}}$, $\Lambda_{\textrm{\MaxProSubSol}}$ and $\Lambda_{\textrm{\MinRemSet}}$.
  For a solution $S\in\MS$ and a nonempty subset $I\subseteq S$, 
  Algorithm~{\rm\ref{alg:linear_delay}}
  enumerates all solutions in $\MS(S,I)$ in
  $O(|S|+\tau_{\textrm{id}}+\tau_{\textrm{\MaxProSubSol}}+\tau_{\textrm{\MinRemSet}})$ delay
  and in $O(|S|(|S|+\sigma_{\textrm{id}}+\sigma_{\textrm{\MaxProSubSol}}+\sigma_{\textrm{\MinRemSet}}))$ space. 
\end{thm}
\begin{proof}
  The procedure \Call{EnumSSD}{$S,I,d$}
  outputs $S$ precisely in its execution.
  If $S$ is \MaxProSubSol-disjoint, which is identified in \lineref{enum_id},
  then all solutions in the second term of the right hand in \eqref{enum_dis}
  are output recursively (\lineref{enum_recur_dis}), and otherwise,
  all solutions in the second term of the right hand in \eqref{enum_no_dis}
  are output recursively (\lineref{enum_recur_no_dis}). These show the correctness.

  We analyze the complexity. Observe that, in the procedure \Call{EnumSSD}{$S,I,d$},
  the input $d$ represents the depth of the search tree, where
  we change the timing when $S$ is output according to the parity of $d$.
  This is called the \emph{alternative method}~\cite{Uno.2003},
  by which the delay can be bounded by the computation time of a single node,
  that is, a single execution of \Call{EnumSSD}{$S,I,d$} ignoring recursive calls.
  We see that the three oracles are called at most once;
  the two for-loops are repeated at most $|S|$ times
  since $\MX$ and $\MY$ are pairwise disjoint, respectively; and
  the processing of $J$ can be done in $O(|S|)$ time
  over the procedure.
  We see that \Call{EnumSSD}{$S,I,d$} can be done
  in $O(|S|+\tau_{\textrm{id}}+\tau_{\textrm{\MaxProSubSol}}+\tau_{\textrm{\MinRemSet}})$ time and
  in $O(|S|+\sigma_{\textrm{id}}+\sigma_{\textrm{\MaxProSubSol}}+\sigma_{\textrm{\MinRemSet}})$ space.
  The height of the search tree is at most $|S|$, and then
  we have the required complexities of delay and space. 
\end{proof}


\paragraph{Proof for \thmref{linear_delay}.}
For a digraph $G=(V,E)$,
the strongly-connected system $(V,\MS_1)$
is an SSD system by \lemref{k_edge_SSD}.
Let $V=\{v_1,v_2,\dots,v_n\}$.
To generate all solutions in $\MS_1$,
we partition $\MS_1$ into $\MS_1(V\setminus\{v_1,v_2,\dots,v_{i-1}\},\{v_i\})$,
$i\in[1,n]$, based on \eqref{basic_partition}.
The solution set $\MS_1(V\setminus\{v_1,v_2,\dots,v_{i-1}\},\{v_i\})$
is not empty since $\{v_i\}$ belongs to this set, and
the only inclusion-maximal solution in 
$\MS_1(V\setminus\{v_1,v_2,\dots,v_{i-1}\},\{v_i\})$
is the strongly-connected component
that contains the vertex $v_i$
in the subgraph $G-\{v_1,v_2,\dots,v_{i-1}\}$.
Denoting by $S_i$ this maximal solution,
we can enumerate all solutions in
$\MS_1(V\setminus\{v_1,v_2,\dots,v_{i-1}\},\{v_i\})$
by executing 
\Call{EnumSSD}{$S_i,\{v_i\},0$} in \algref{linear_delay}.
The correctness follows by \thmref{linear_delay_SSD}. 

We can generate each $S_i$, $i\in[1,n]$ in $O(n+m)$ time. 
The three oracles for $(V,\MS_1)$ can be implemented
so that the time complexity is linear,
by which the theorem is proved by \thmref{linear_delay_SSD}. 
\begin{description}
\item[$\Lambda_{\textrm{id}}(S)$:]
  If $|S|\ge2$, then run
  \algref{MRS_determine_disjointness} with
  $G[S]$ as the input.
  Otherwise, it suffices to output just \True. 
  The computation time is $O(n+m)$ by
  \lemref{str_MRS_determine_disjoint}.
\item[$\Lambda_{\textrm{\MaxProSubSol}}(S,I)$:]
  If $|S|\ge2$, then
  run \algref{MPS_disjoint}
  with $G[S]$ as the input.
  Otherwise, it suffices to output $\emptyset$. 
  The computation time is $O(n+m)$ by
  \lemref{str_MRS_not_disjoint}.
\item[$\Lambda_{\textrm{\MinRemSet}}(S,I)$:]
  If $|S|\ge2$, then run \algref{MRS_disjoint}
  with $G[S]$ as the input
  and extract all generated \MinRemSets.
  Otherwise, it suffices to output $\emptyset$. 
  The computation time is $O(n+m)$ by
  \lemref{str_MRS_disjoint}.
\end{description}
\hfill\QED

\section{A Novel Sufficient Condition on Existence of Hamiltonian Cycles}
\label{sec:hamilton}

In this section, we show a proof for \thmref{hamilton} that
states a sufficient condition on the existence of a Hamiltonian cycle
in a digraph. 

Let $G=(V,E)$ be a strongly-connected \MaxProSubSol-disjoint graph.
Let $\MX_G:=\MPSS_{G,1}(V)$. 
We define an auxiliary digraph $H_G\triangleq(\MX_G,\MY_G)$, where
$\MY_G$ is the set of arcs that is defined to be
\[
\MY_G\triangleq
\{(X,X')\in\MX_G\times\MX_G\mid X\ne X',\ \exists u\in X,\ \exists u'\in X',\ (u,u')\in E\}. 
\]

\begin{lem}
  \label{lem:hamilton}
  Let $G=(V,E)$ be a strongly-connected \MaxProSubSol-disjoint graph such that $n=|V|\ge2$. 
  Let $q:=|\MX_G|$. 
  \begin{description}
    \item[\rm (i)] There is no simple $\ell$-cycle in $H_G$ for any $2\le\ell<q$. 
    \item[\rm (ii)] There is a Hamiltonian cycle in $H_G$. 
    \item[\rm (iii)] There is no simple $p$-cycle in $G$ for any $2\le p<n$
      that visits at least one vertex for every \MaxProSubSol in $\MX_G$. 
    \item[\rm  (iv)] There is a Hamiltonian cycle in $G$. 
  \end{description}
\end{lem}
\begin{proof}
  Let $\MX_G:=\{X_1,X_2,\dots,X_q\}$. Observe that $q\ge2$ holds by \lemref{num_MRSs}(ii). 

  \noindent
      (i) If there is a simple $\ell$-cycle,
      say $X_1\to X_2\to\dots\to X_\ell\to X_1$,
      then $X_1\cup X_2\cup\dots\cup X_\ell\subsetneq V$
      would induce a strongly-connected
      subgraph of $G$, contradicting the maximality of $X_1,X_2,\dots,X_\ell$.

  \noindent
      (ii) If $H_G$ is acyclic, then $G$ would not be strongly-connected.
      The possible length of a simple cycle in $H_G$
      is $q$ by (i). 

  \noindent
      (iii) Suppose that such a simple $p$-cycle exists in $G$.
      Let $X'_i\subseteq X_i$, $i\in[1,q]$ denote the nonempty
      subset of vertices in $X_i$
      that are visited by the simple $p$-cycle,
      where $X'_j\subsetneq X_j$ holds for some $j\in[1,q]$ by $p<n$.       
      We see that $X':=X'_1\cup X'_2\cup\dots\cup X'_q$ induces
      a strongly-connected subgraph of $G$.
      There exists a \MaxProSubSol $Z$ of $V$ such that $Z\supseteq X'$.
      Let $Z_i:=Z\cap X_i$, $i\in[1,q]$,
      where $X'_i\subseteq Z_i\subseteq X_i$ holds.
      There is $j'\in[1,q]$ such that $Z_{j'}\subsetneq X_{j'}$
      since $Z$ is a PSS of $V$.
      However, such $Z$ would intersect $X_{j'}$,
      contradicting that $V$ is MaxPSS-disjoint. 
      
  \noindent
      (iv) By (ii), there is a simple $q$-cycle in $H_G$.
      From the $q$-cycle in $H_G$, 
      we can construct a simple cycle in $G$ that visits
      at least one vertex of every $X_1,X_2,\dots,X_q$.      
      The length of the simple cycle must be $n$ by (iii). 
\end{proof}

\paragraph{Proof for \thmref{hamilton}.}
It is immediate by \lemref{hamilton}(iv)
since, if $G=(V,E)$ is Hamiltonian,
any supergraph $\hat{G}=(V,\hat{E})$
for $E\subseteq\hat{E}\subseteq V\times V$
is also Hamiltonian. 
\QED

In the reminder of this section,
let us observe how \MaxProSubSol-disjointness or \MinRemSet-disjointness
changes by adding edges to/deleting edges from a strongly-connected digraph. 
\propref{mono} shows the monotonicity of these two disjointnesses
for $k$-edge-connectivity. 
\begin{prop}
  \label{prop:mono}
  Let $G=(V,E)$ be a $k$-edge-connected graph such that $|V|\ge2$.
  \begin{itemize}
  \item[\rm (i)]  If $G$ is \MinRemSet-disjoint,
    then $G+e'$ is \MinRemSet-disjoint
    for any $e'\in\bar{E}$.
  \item[\rm (ii)] If $G$ is \MaxProSubSol-disjoint,
    then $G-e$ is \MaxProSubSol-disjoint
    for any $e\in E$ such that $G-e$ is $k$-edge-connected. 
  \end{itemize}
\end{prop}
\begin{proof}
  (i)
  Let $G'=(V,E\cup\{e'\})$. Clearly, $G'$ is $k$-edge-connected.
  For any $Y_1\in\MRS_{G,k}(V)$, we see that $(G+e')-Y_1$ is $k$-edge-connected
  since $G-Y_1$ is $k$-edge-connected. Hence $Y_1$ is an \RemSet of $G'$,
  and there exists $Y_1'\in\MRS_{G',k}(V)$ such that $Y_1'\subseteq Y_1$.
  By \lemref{num_MRSs}(ii), there is $Y_2\in\MRS_{G,k}(V)$ such that $Y_2\ne Y_1$.
  Similarly, there exists $Y_2'\in\MRS_{G',k}(V)$ such that $Y_2'\subseteq Y_2$.
  The graph $G$ is \MinRemSet-disjoint, and hence $Y'_1\cap Y'_2\subseteq Y_1\cap Y_2=\emptyset$ holds,
  as required.
  (ii) is immediate from (i). 
\end{proof}

A strongly-connected \MaxProSubSol-disjoint digraph can be turned
into \MinRemSet-disjoint by adding an edge. 
\begin{prop}
  \label{prop:add_MaxPSS}
  Let $G=(V,E)$ be a strongly-connected \MaxProSubSol-disjoint digraph such that $|V|\ge 3$.
  There is $e'\in\bar{E}$ such that $G+e'$ is \MinRemSet-disjoint.
\end{prop}
\begin{proof}
  In $G$, there is a Hamiltonian cycle by \thmref{hamilton},
  say $v_1\to v_2\to\dots\to v_n\to v_1$.
  For $i,j\in[1,n]$, no edge $(v_i,v_j)$ exists if $i+2\le j$
  by \lemref{hamilton}(iii).
  Adding such $(v_i,v_j)$ to $G$ makes the graph \MinRemSet-disjoint. 
\end{proof}

Concerning \propref{add_MaxPSS}, we can find an edge $e'$ to be added
in linear time by constructing the auxiliary digraph $H_G$ and examining it.

Let us observe Propositions~\ref{prop:mono} and \ref{prop:add_MaxPSS}
from another point of view. Let $G=(V,E)$ be a strongly-connected digraph
and $F\subseteq E$ be a subset such that $(V,F)$ is strongly-connected.
\begin{itemize}
\item If $(V,F)$ is \MinRemSet-disjoint, then $(V,F')$ is \MinRemSet-disjoint
  for any $F\subseteq F'\subseteq E$ (\propref{mono}(i)). 
\item If $(V,F)$ is \MaxProSubSol-disjoint, then there is an edge  $e\in E\setminus F$
  such that $(V,F\cup\{e\})$ is \MinRemSet-disjoint or no such $e$ exists (\propref{add_MaxPSS}).
  In the latter case, $G$ is \MaxProSubSol-disjoint. 
\end{itemize}
Then can we make a strongly-connected \MinRemSet-disjoint digraph
into \MaxProSubSol-disjoint by deleting edges?
Unfortunately, as expected, this problem is computationally hard. 

\begin{prop}
  \label{prop:NPhard}
  Let $G=(V,E)$ be a strongly-connected \MinRemSet-disjoint digraph.
  There is no polynomial-time algorithm for deciding whether
  there is $F\subsetneq E$ such that $(V,F)$ is
  a strongly-connected \MaxProSubSol-disjoint digraph unless $\mathcal{P}=\mathcal{NP}$. 
\end{prop}
\begin{proof}
  We show that the answer to the decision problem is yes if and only if $G$ is Hamiltonian.
  Suppose that the answer is yes, and let $F\subseteq E$ be a yes-certificate.
  Then $(V,F)$ is \MaxProSubSol-disjoint and hence $G$ is Hamiltonian (\thmref{hamilton}).
  Conversely, suppose that $G$ is Hamiltonian,
  and let $F$ be the edge set that constitutes a Hamiltonian cycle.
  Then $F$ is a yes-certificate for the decision problem since
  $(V,F)$, a directed simple cycle, is obviously \MaxProSubSol-disjoint. 
\end{proof}

\section{Concluding Remarks}
\label{sec:conc}
In this paper, we introduced SSD (Superset-Subset-Disjoint) set system.
Based on \thmref{disjoint},
which states that every solution in an SSD system
is \MaxProSubSol-disjoint and/or \MinRemSet-disjoint, 
we obtained the following results;
(\thmref{partition}) for any $k\in\bbZ_+$ and graph $G$,
there exists a unique graph decomposition based on $\MS_{G,k}$-\MaxProSubSols. 
(\thmref{strong}) For a strongly-connected digraph $G$,
we can decide whether $G$ is \MaxProSubSol-disjoint or not in linear time
and generate all \MaxProSubSols in linear time. 
(\thmref{linear_delay}) We can enumerate all solutions
in a strongly-connected system of a given digraph in linear delay. 
(\thmref{hamilton}) A digraph $G$ is Hamiltonian
if there is a spanning subgraph that is strongly-connected and 
\MaxProSubSol-disjoint. 

Here is the list of future work:
\begin{itemize}
\item For \thmref{strong}, we developed a linear-time algorithm
  that generates \MaxProSubSols of a solution in
  the $k$-edge-connected system $\MS_{G,k}$
  when $k=1$ and $G=(V,E)$ is directed.
  It would be interesting to study other cases,
  i.e., $k\ge 2$ or $G$ is undirected.
  The problem is trivial if $k=1$ and $G$ is undirected;
  a \MaxProSubSol of $G$ is $V\setminus\{v\}$ for a vertex $v\in V$ that
  is not an articulation point.
  The case when $k=2$ and $G$ is undirected
  was already studied in \cite{TH.2023}.
  The $k$-edge-connected system in a weighted graph is also of our interest. 
\item The possibility of SSD system as well as SD system
  should be explored. A natural question
  is to ask whether $k$-vertex-connected system is SSD or not,
  where it was shown confluent~\cite{Haraguchi.2022}. 
  Our results on $k$-edge-connected system
  are based on that the system is SSD and confluent.
  If $k$-vertex-connected system is SSD, then
  similar results as $k$-edge-connected system
  would be obtained. 
\item The graph decomposition and
  the sufficient condition on the existence of Hamiltonian cycles 
  are interesting by themselves.
  We leave as future work
  precise comparisons with existing ones
  and exploration of their applications to other graph problems. 
\end{itemize}

\clearpage
\bibliographystyle{plainurl}
\bibliography{ref}


\end{document}